\documentclass[11pt]{article}

\newtheorem{theorem}{Theorem}[section]

\newtheorem{proposition}[theorem]{Proposition}
\newtheorem{definition}[theorem]{Definition}
\newtheorem{corollary}[theorem]{Corollary}

\newcommand{\sq}{\hbox{\rlap{$\sqcap$}$\sqcup$}}
\newcommand{\qed}{\hspace*{\fill}\sq}
\newenvironment{proof}{\noindent {\bf Proof.}\ }{\qed\par\vskip 4mm\par}

\usepackage[cmex10]{amsmath}
\usepackage{amssymb,amsfonts,textcomp}
\usepackage{latexsym}
\usepackage[nospace, compress]{cite}
\usepackage[noend,ruled]{algorithm2e}
\usepackage{times} 
\usepackage{graphicx} 
\usepackage{latexsym}
\usepackage{mathtools}
\usepackage{subfigure}
\usepackage{svn}
\usepackage{bm}
\usepackage{cite}
\usepackage{paralist}
\usepackage[defaultlines=2, all]{nowidow}

\allowdisplaybreaks
\newcommand{\fixme}[1]{}

\newcommand{\reals}{{{\mathbb{R}}}}

\newcommand{\timemoments}{{{\mathbb{T}}}}
\newcommand{\orgs}{{{\mathcal{O}}}}
\newcommand{\schedules}{{{\Gamma}}}
\newcommand{\powjobs}{{{\mathfrak{J}}}}
\newcommand{\orders}{{{\mathcal{L}}}}
\newcommand{\np}{{\mathrm{NP}}}
\newcommand{\p}{{\mathrm{P}}}
\newcommand{\expected}{{{\mathbb{E}}}}
\newcommand{\probability}{{{\mathrm{P}}}}

\begin{document}

\title{Non-monetary fair scheduling---a cooperative game theory approach}

\author{Piotr Skowron \\
   University of Warsaw \\
   Email: p.skowron@mimuw.edu.pl \\
   \and
   Krzysztof Rzadca \\
   University of Warsaw \\
   Email: krzadca@mimuw.edu.pl
   }

\date{}

\maketitle

%=========================================================================
%  Abstract
%=========================================================================

\begin{abstract}
We consider a multi-organizational system in which each organization contributes processors to the global pool but also jobs to be processed on the common resources. The fairness of the scheduling algorithm is essential for the stability and even for the existence of such systems (as organizations may refuse to join an unfair system).

We consider on-line, non-clairvoyant scheduling of sequential jobs.
The started jobs cannot be stopped, canceled, preempted, or moved to other processors.
We consider identical processors, but most of our results can be extended to related or unrelated processors.

We model the fair scheduling problem as a cooperative game and we use the Shapley value to determine the ideal fair schedule. In contrast to the previous works, we do not use money to assess the relative utilities of jobs. Instead, to calculate the contribution of an organization, we determine how the presence of this organization influences the performance of other organizations. Our approach can be used with arbitrary utility function (e.g., flow time, tardiness, resource utilization), but we argue that the utility function should be strategy resilient. The organizations should be discouraged from splitting, merging or delaying their jobs. We present the unique (to within a multiplicative and additive constants) strategy resilient utility function.

We show that the problem of fair scheduling is NP-hard and hard to approximate. However, we show that the problem parameterized with the number of organizations is fixed parameter tractable (FPT). Also, for unit-size jobs, we present a fully polynomial-time randomized approximation scheme (FPRAS). Although for the large number of the organizations the problem is computationally hard, the presented exponential algorithm can be used as a fairness benchmark.

All our algorithms are greedy, i.e., they don't leave free processors if there are waiting jobs.
We show that any greedy algorithm results in at most 25\% loss of the resource utilization in comparison with the globally optimal algorithm. As a corollary we conclude that the resource underutilization, being the result of the fairness requirement, is (tightly) upper bounded by $0.25$.

We propose a heuristic scheduling algorithm for the fair scheduling problem. We experimentally evaluate the heuristic and compare its fairness to fair share, round robin and the exact exponential algorithm. Our results show that fairness of the heuristic algorithm is close to the optimal. The difference between our heuristic and the fair share algorithm is more visible on longer traces with more organizations. These results show that assigning static target shares (as in the fair share algorithm) is not fair in multi-organizational systems and that instead dynamic measures of organizations' contributions should be used.
\end{abstract}

\bigskip

\centerline{{\bf Keywords}: fair scheduling, cooperative game theory, Shapley value, multi-organization,}
\centerline{cooperation, strategy resistance, approximation algorithms, inapproximability, fairness.}

%=========================================================================
%  Introduction
%=========================================================================

\section{Introduction}\label{sec::introduction}
In multi-organizational systems, participating organizations give access to their local resources; in return their loads can be processed on other resources. The examples of such systems include PlanetLab\footnote{www.planet-lab.org/}, grids (Grid5000\footnote{www.grid5000.fr}, EGEE\footnote{egee-technical.web.cern.ch}) or organizationally distributed storage systems~\cite{Hasan:2005:SPS:1058431.1059197}.
There are many incentives for federating into consortia: the possibility of decreasing the costs of management and maintenance (one large system can be managed more efficiently than several smaller ones), but also the willingness to utilize resources more efficiently. Peak loads can be offloaded to remote resources. Moreover, organizations can access specialized resources or the whole platform (which permits e.g. testing on a large scale). 

In the multi-organizational and multi-user systems fairness of the resource allocation mechanisms is equally important as its efficiency. Efficiency of BitTorrent depends on users' collaboration, which in turn requires the available download bandwidth to be distributed fairly~\cite{Rahman:2011:DSA:2018436.2018458}. Fairness has been also discussed in storage systems~\cite{proportionalScheduling2007Fast, adaptiveOverload2003USENIX, fairScheduling2003IEEE, Gulati2008, parda2009Fast, 2levelIOScheduling, multiDimensionalScheduling} and computer networks~\cite{Wierman}. In scheduling, for instance, a significant part of the description of Maui~\cite{Jackson:2001:CAM:646382.689682}, perhaps the most common cluster scheduler, focuses on the fair-share mechanism.
%% In all these systems fairness is essential
%% for the system to be joined by the organizations, and so fairness is a basic requirement for the existence of these systems.
Nevertheless there is no universal agreement on the meaning of fairness; next, we review approaches most commonly used in literature: distributive fairness and game theory.

% TODO: Bi-criteria algorithm for scheduling jobs on cluster platforms.

In distributive fairness organizations are ensured a fraction of the resources according to  \emph{predefined (given) shares}. The share of an organization may depend on the perceived importance of the workload, payments~\cite{proportionalScheduling2007Fast, parda2009Fast, fairScheduling2003IEEE, Gulati2008}; or calculated to satisfy (predefined) service level agreements~\cite{2levelIOScheduling, multiDimensionalScheduling, Triage2005ACM}. The literature on distributive fairness describes algorithms distributing resources according to the given shares, but does not describe how the shares should be set. In scheduling, distributive fairness is implemented through fair queuing mechanism: YFQ \cite{yfqScheduling}, SFQ and FSFQ \cite{simpleSfqScheduling, fsfqScheduling}, or their modifications~\cite{proportionalScheduling2007Fast, adaptiveOverload2003USENIX, fairScheduling2003IEEE, generalAboutControlTheory, Gulati2008, parda2009Fast, 2levelIOScheduling, multiDimensionalScheduling}. 

A different approach is to optimize directly the performance (the utility) of users, rather than just the allocated resources.  \cite{journals/eor/KostrevaOW04} proposes an axiomatic characterization of fairness based on multi-objective optimization; \cite{Rzadca2007FairGameTheoreticResource} applies this concept to scheduling in a multi-organizational system. Inoie~et~al.~\cite{Inoie04paretoset} proposes a similar approach for load balancing: a fair solution must be Pareto-optimal and the revenues of the players must be proportional to the revenues in Nash equilibrium.

While distributive fairness might be justified in case of centrally-managed systems (e.g. Amazon EC2 or a single HPC center), in our opinion it is inappropriate for consortia (e.g., PlanetLab or non-commercial scientific systems like Grid5000 or EGEE) in which there is no single ``owner'' and the participating organizations may take actions (e.g. rescheduling jobs on their resources, adding local resources, or isolating into subsystems). In such systems the shares of the participating organizations should depend both on their workload and on the owned resources; intuitively an organization that contributes many ``useful'' machines should be favored; similarly an organization that has only a few jobs.

If agents may form binding agreements, cooperative game theory studies the stability of resulting agreements (coalitions and revenues).
The Shapley value~\cite{shapley_53}, the established solution concept that characterizes what is a \emph{fair} distribution of the total revenue of the coalition between the participating agents, has been used in scheduling theory. However, all the models we are aware of use the concept of money.
The works of Carroll et at.~\cite{divisibleLoadSchedCoalitionalGames}, Mishra et al.~\cite{costSharingInJobSched}, Mashayekhy and Grosu~\cite{Mashayekhy:2011:MMD:2358951.2359134} and Moulin et al.~\cite{Moulin:2007:SFP:1527888.1527890} describe algorithms and the process of forming the coalitions for scheduling. These works assume that each job has a certain \emph{monetary} value for the issuing organization and each organization has its initial monetary budget.

Money may have negative consequences on the stakeholders of resource-sharing consortia.  Using (or even mentioning) money discourages people from cooperating~\cite{citeulike:1087789}. This stays in sharp contrast with the idea behind the academic systems --- sharing the infrastructure is a step towards closer cooperation. Additionally, we believe that using money is inconvenient in non-academic systems as well. In many contexts, it is not clear how to valuate the completion of the job or the usage of a resource (especially when workload changes dynamically).
We think that the accurate valuation is equally important (and perhaps equally difficult) as the initial problem of fair scheduling.
Although auctions~\cite{buyya05ge} or commodity markets~\cite{kenyon04gr} have been proposed to set prices, these approaches implicitly require to set the reference value to determine profitability.
Other works on \emph{monetary} game-theoretical models for scheduling include~\cite{Ghosh20051366, 10.1109/TPDS.2007.34, Penmasta:2006:PUJ:1898699.1898880, Grosu:2003:TMF:824470.825337, Grosu04auction-basedresource}; monetary approach is also used for other resource allocation problems, e.g. network bandwidth allocation~\cite{Yaiche:2000:GTF:355154.355166}. However, none of these works describes how to valuate jobs and resources. 

In a non-monetary approach proposed by Dutot~el~al.~\cite{Trystram2010ApproximationAlgorithms} the jobs are scheduled to minimize the global performance metric (the makespan) with an additional requirement --- the utility of each player cannot be worse than if the player would act alone. Such approach ensures the stability of the system against actions of any single user (it is not profitable for the user to leave the system and to act alone) but not to the formation of sub-coalitions. 

In the selfish job model~\cite{vocking-selfishloadbalancing} the agents are the jobs that selfishly choose processors on which to execute. Similarly to our model the resources are shared and treated as common good; however, no agent contributes resources.

An alternative to scheduling is to allow jobs to share resources concurrently. In congestion games~\cite{routingGames, finiteCongestionGames, linearCongestionGames} the utility of the player using a resource $R$ depends on the number of the players concurrently using $R$; the players are acting selfishly. Congestion games for divisible load scheduling were analyzed by Grosu~et.~al~\cite{Grosu02agame-theoretic}.

In this paper we propose fair scheduling algorithms for systems composed of multiple organizations (in contrast to the case of multiple organizations using a system owned by a single entity). We model the organizations, their machines and their jobs as a cooperative game. In this game we do not use the concept of money. When measuring the contribution of the organization $O$ we analyze how the presence of $O$ in the grand coalition influences the completion times of the jobs of all participating organization. This contribution is expressed in the same units as the utility of the organization. In the design of the fair algorithm we use the concept of Shapley value. In contrast to simple cooperative game, in our case the value of the coalition (the total utility of the organizations in this coalition) depends on the underlying scheduling algorithm. This makes the problem of calculating the contributions of the organizations more involved.
First we develop algorithms for arbitrary utilities (e.g. resource utilization, tardiness, flow time, etc.). Next we argue that designing the scheduling mechanism itself is not enough; we show that the utility function must be chosen to discourage organizations from manipulating their workloads (e.g. merging or splitting the jobs --- similar ideas have been proposed for the money-based models~\cite{{Moulin:2007:SFP:1527888.1527890}}).
We present an exponential scheduling algorithm for the strategy resilient utility function. We show that the fair scheduling problem is NP-hard and difficult to approximate. For a simpler case, when all the jobs are unit-size, we present a fully polynomial-time randomized approximation scheme (FPRAS). According to our experiments this algorithm is close to the optimum when used as a heuristics for workloads with different sizes of the jobs.

\textbf{Our contribution is the following:}
\begin{enumerate}
\item We derive the definition of the fair algorithm from the cooperative game theory axioms (Definitions~\ref{def::fairScheduling}~and~\ref{def::fairAlgorithm}, Algorithm~\ref{alg:fairAlg} and Theorem~\ref{thm::fairnessInGeneral}). The algorithm uses only the notions regarding the performance of the system (no money-based mechanisms).
\item We present the axioms (Section~\ref{sec::strategyProof}) and the definition of the fair utility function (Theorem~\ref{thm:strategyProofUtility}) --- this function is similar to the flow time metric but the differences make it strategy-resilient (Proposition~\ref{thm::metricLikeFlowTime}).
\item We show that the fair scheduling problem is NP-complete (Theorem~\ref{thm::contributionsNp}) and hard to approximate (Theorem~\ref{thm::inapprox}). However, the problem parameterized with the number of organizations is fixed parameter tractable (FPT).
\item We present an FPRAS for a special case with unit-size jobs (Algorithm~\ref{alg:fairAlgUnitSize}, Theorems~\ref{thm::approxRandomized}~and~\ref{thm:fpras}).
\item We show the tight bounds on the resource underutilization due to the fairness of the algorithm. Our result is even more general and applies to all greedy algorithms (i.e. such algorithms that at every time moment in which there is a free processor and a non-empty set of ready, but not scheduled jobs, schedules some job on some free processor).
\item We propose a practical heuristic that schedules jobs according to an estimated Shapley value. The heuristic estimates the contribution of an organization by the number of CPU-timeunits an organization contributes for computing jobs of other organizations; the algorithm schedules the jobs to minimize the maximal difference between the utility and the contribution over all organizations.
\item Finally, we conduct simulation experiments to verify fairness of commonly-used scheduling algorithms (Section~\ref{sec::experiments}). The experiments show that although the fair share algorithm is considerably better than round robin (which does not aim to optimize fairness), our heuristic constantly outperforms fair share, being close to the optimal algorithm and the randomized approximation algorithm. The main conclusion from the experimental part of this paper is that ensuring that each party is given a fair share of resources (the distributive fairness approach) might not be sufficient in systems with dynamic job arrival patterns. An algorithm based on the Shapley value, that explicitly considers the organization's impact on other organizations' utilities, produces more fair schedules.
\end{enumerate}

In this paper we use very mild assumptions about the jobs. We do not require to know their valuations, durations, or their future incoming pattern. Thus, we believe the presented results have practical consequences for real-life job schedulers. Also, our exponential algorithm forms a benchmark for comparing the fairness of other polynomial-time scheduling algorithms. The experimental comparison of some real scheduling algorithms suggests that the polynomial-time heuristic algorithms inspired by the ones presented in this paper often result in a better fairness than the currently most popular fair share algorithm.

\section{Preliminaries}\label{sec:prelim}

\textbf{Organizations, machines, jobs}. We consider a system built by a set of independent \emph{organizations} $\orgs = \{O^{(1)}, O^{(2)}, \dots O^{(k)}\}$.
Each organization $O^{(u)}$ owns a \emph{computational cluster} consisting of $m^{(u)}$ \emph{machines} (\emph{processors}) denoted as $M^{(u)}_{1}, M^{(u)}_{2}, \dots M^{(u)}_{m^{(u)}}$ and 
produces its \emph{jobs}, denoted as $J^{(u)}_{1}, J^{(u)}_{2}, \dots$. Each job $J^{(u)}_{i}$ has \emph{release time} $r^{(u)}_{i} \in \timemoments$, where $\timemoments$ is a discrete set of time moments.
We consider an on-line problem in which each job is unknown until its release time.
We consider a non-clairvoyant model i.e., the job's processing time is unknown until the job completes (hence we do not need to use imprecise~\cite{lee2006user} run-time estimates).
%\fixme{fuzzy} Our approach could also be used for the clairvoyant model; it would allow us to try to be fair in future time moments; with non-clairvoyant model we focus on the current time moment.
For the sake of simplicity of the presentation we assume that machines are identical, i.e. each job $J^{(u)}_{i}$ can be executed at any machine and its processing always takes $p^{(u)}_{i}$ time units; $p^{(u)}_{i}$ is the \emph{processing time}. Most of the results, however, can be extended to the case of related machines, where $p^{(u)}_{i}$ is a function of the schedule -- the only exception make the results in Section~\ref{sec::unitSize}, where we rely on the assumption that each job processed on any machine takes exactly one time unit. The results even generalize to the case of unrelated machines, however if we assume non-clairvoyant model with unrelated machines (i.e., we do not know the processing times of the jobs on any machine) then we cannot optimize the assignment of jobs to machines. 
%; which leads to producing globally inefficient schedules. 

The jobs are sequential (this is a standard assumption in many scheduling models and, particularly, in the selfish job model~\cite{vocking-selfishloadbalancing}; an alternative is to consider the parallel jobs, which we plan to do in the future).
Once a job is started, the scheduler cannot preempt it or migrate it to other machine (this assumption is usual in HPC scheduling because of high migration costs).
Finally, we assume that the jobs of each individual organization should be started in the order in which they are presented. This allows organizations to have an internal prioritization of their jobs.\medskip

\textbf{Cooperation, schedules}. Organizations can cooperate and share their infrastructure; in such case we say that organizations form a \emph{coalition}. Formally, a coalition $\mathcal{C}$ is a subset of the set of all organizations, $\mathcal{C} \subseteq \orgs$. We also consider a specific coalition consisting of all organizations, which we call a \emph{grand coalition} and denote as $\mathcal{C}_{g}$ (formally, $\mathcal{C}_{g} = \orgs$, but in some contexts we use the notation $\mathcal{C}_{g}$ to emphasize that we are referring to the set of the organizations that cooperate).
The coalition must agree on the schedule $\sigma = \bigcup_{(u)} \bigcup_{i}\{(J^{(u)}_{i}, s^{(u)}_{i}, M(J^{(u)}_{i}))\}$ which is a set of triples; a triple $(J^{(u)}_{i}, s^{(u)}_{i}, M(J^{(u)}_{i}))$ denotes a job $J^{(u)}_{i}$ started at time moment $s^{(u)}_{i} \geq r^{(u)}_{i}$ on machine $M(J^{(u)}_{i})$. We assume that a machine executes at most one job at any time moment.
We often identify a job $J^{(u)}_{i}$ with a pair $(s^{(u)}_{i}, p^{(u)}_{i})$; and a schedule with $\bigcup_{(u)} \bigcup_{i}\{(s^{(u)}_{i}, p^{(u)}_{i}))\}$ (we do so for a more compact presentation of our results).
The coalition uses all the machines of its participants and schedules consecutive tasks on available machines.
We consider only greedy schedules: at any time moment if there is a free processor and a non-empty set of ready, but not scheduled jobs, some job must be assigned to the free processor.
Since we do not know neither the characteristics of the future workload nor the duration of the started but not yet completed jobs, any non-greedy policy would result in unnecessary delays in processing jobs. Also, such greedy policies are used in real-world schedulers~\cite{Jackson:2001:CAM:646382.689682}.

Let $\powjobs$ denote the set of all possible sets of the jobs. An \emph{online scheduling algorithm} (in short a scheduling algorithm) $\mathcal{A}: \powjobs \times \timemoments \rightarrow \orgs$ is an online algorithm that continuously builds a schedule: for a given time moment $t \in \timemoments$ such that there is a free machine in $t$ and a set of jobs released before $t$ but not yet scheduled: $\mathcal{J} \in \mathfrak{J}$, $\mathcal{A}(\mathcal{J}, t)$ returns the organization the task of which should be started.
The set of all possible schedules produced by such algorithms is the set of \emph{feasible schedules} and denoted by $\schedules$. We recall that in each feasible schedule the tasks of a single organization are started in a FIFO order. \medskip

\textbf{Objectives}. We consider a \emph{utility function} $\psi: \schedules \times \orgs \times \timemoments \rightarrow \reals$ that for a given schedule $\sigma \in \schedules$, an organization $O^{(u)}$, and a time moment $t$ gives the value corresponding to the $O^{(u)}$ organization's satisfaction from a schedule $\sigma$ until $t$.
The examples of such utility functions that are common in scheduling theory are: flow time, resource utilization, turnaround, etc. Our scheduling algorithms will only use the notions of the utilities and do not require any external payments.
%We assume that $\psi$ is anonymous, that is for any time moment $t$ and any two organizations $O^{(u_1)}$ and $O^{(u_2)}$ and any two schedules $\sigma_1$ and $\sigma_2$, such that the starting and completion times of the tasks started before $t$ of $O^{(u_1)}$ in $\sigma_1$ are the same as of the tasks of $O^{(u_2)}$ in $\sigma_2$, then:
%$\psi(\sigma_1, O^{(u_1)}, t) = \psi(\sigma_2, O^{(u_2)}, t)$.
%Without anonymity some organizations would be favored based only on their identity. Anonymity is also a standard assumption in TODO.

Since a schedule $\sigma$ is fully determined by a scheduling algorithm $\mathcal{A}$ and a coalition of organizations $\mathcal{C}$, we often identify $\psi(\mathcal{A}, \mathcal{C}, O^{(u)}, t)$ with appropriate $\psi(\sigma, O^{(u)}, t)$. Also, we use a shorter notation $\psi^{(u)}(\mathcal{C})$ instead of $\psi(\mathcal{A}, \mathcal{C}, O^{(u)}, t)$ whenever the $\mathcal{A}$ and $t$ are known from the context.
We define the \emph{characteristic function} $v: \schedules \times \timemoments \rightarrow \reals$ describing the total utility of the organizations from a schedule:
% \begin{align*}
$v(\mathcal{A}, \mathcal{C}, t) = \sum_{O^{(u)} \in \mathcal{C}} \psi(\mathcal{A}, \mathcal{C}, O^{(u)}, t) \textrm{.}$
% \end{align*}
Analogously as above, we can use an equivalent formulation:\\
%\begin{align*}
$v(\sigma, t) =\sum_{O^{(u)} \in \mathcal{C}} \psi(\sigma, O^{(u)}, t) \textrm{,}$
%\end{align*}
also using a shorter notations $v(\mathcal{C})$ whenever it is possible.
Note that the utilities of the organizations $\psi^{(u)}(\mathcal{C})$ constitute a division of the value of the coalition $v(\mathcal{C})$.

\section{Fair scheduling based on the Shapley value}
In this section our goal is to find a scheduling algorithm $\mathcal{A}$ that in each time moment $t$ ensures a fair distribution of the value of the coalition $v(\mathcal{C})$ between the participating organizations. We will denote this desired fair division of the value $v$ as $\phi^{(1)}(v), \phi^{(2)}(v), \dots, \phi^{(k)}(v)$ meaning that $\phi^{(u)}(v)$ denotes the ideally fair revenue (utility) obtained by organization $O^{(u)}$. 
% For defining what a fair distribution of $v(\mathcal{C})$ actually means, we use solution concepts from cooperative game theory.
We would like the values $\phi^{(u)}(v)$ to satisfy the fairness properties, first proposed by Shapley~\cite{shapley_53} (below we give intuitive motivations; see~\cite{shapley_53} for further arguments).

\begin{enumerate}
\item efficiency -- the total value $v(\mathcal{C})$ is distributed:
\begin{align*}
\sum_{O^{(u)} \in \mathcal{C}}\phi^{(u)}(v(\mathcal{C})) = v(\mathcal{C}) \textrm{.}
\end{align*}
\item symmetry -- the organizations $O^{(u)}$ and $O^{(u')}$ having indistinguishable contributions obtain the same profits:
\begin{align*}
\left(\forall_{\mathcal{C'} \subset \mathcal{C}:  O^{(u)}, O^{(u')} \notin \mathcal{C'}}\;
v(\mathcal{C'} \cup \{O^{(u)}\}) = v(\mathcal{C'} \cup \{O^{(u')}\})\right)
\Rightarrow \phi^{(u)}(v(\mathcal{C})) = \phi^{(u')}(v(\mathcal{C})) \textrm{.}
\end{align*}
\item additivity -- for any two characteristic functions $v$ and $w$ and a function ($v$+$w$): $\forall_{\mathcal{C'} \subseteq \mathcal{C}}$ ($v$+$w$)$(\mathcal{C'}) = v(\mathcal{C'}) + w(\mathcal{C'})$ we have that $\forall_{\mathcal{C'} \subseteq \mathcal{C}}\,\forall_{u}$:
\begin{align*}
\phi^{(u)}((v\textrm{+}w)(\mathcal{C})) = \phi^{(u)}(v(\mathcal{C})) + \phi^{(u)}(w(\mathcal{C})) \textrm{.}
\end{align*}
Consider any two independent schedules $\sigma_1$ and $\sigma_2$  that together form a schedule $\sigma_3 = \sigma_1 \cup \sigma_2$ ($\sigma_1$ and $\sigma_2$ are independent iff removing any subset of the jobs from $\sigma_1$ does not influence the completion time of any job in $\sigma_2$ and vice versa). The profit of an organization that participates only in one schedule (say $\sigma_1$) must be the same in case of $\sigma_1$ and $\sigma_3$ (intuitively: the jobs that do not influence the current schedule, also do not influence the current profits). The profit of every organization that participates in both schedules should in $\sigma_3$ be the sum of the profits in $\sigma_1$ and $\sigma_2$. Intuitively: if the schedules are independent then the profits are independent too.
\item dummy -- an organization that does not increase the value of any coalition $C' \subset C$ gets nothing:
\begin{align*}
\left(\forall_{\mathcal{C'} \subset \mathcal{C}}: v(\mathcal{C'} \cup \{O^{(u)}\}) = v(\mathcal{C'})\right)
\Rightarrow \phi^{(u)}(v(\mathcal{C})) = 0 \textrm{.}
\end{align*}
\end{enumerate}

Since the four properties are actually the axioms of the Shapley value~\cite{shapley_53}, they fully determine the single mapping between the coalition values and the profits of organizations (known as the Shapley value). In game theory the Shapley value is considered the classic mechanism ensuring the fair division of the revenue of the coalition\footnote{The Shapley value has other interesting axiomatic characterizations~\cite{aumann_shapley}.}. The Shapley value can be computed by the following formula~\cite{shapley_53}:
\begin{equation}\label{eq::shapleyValue}
\begin{aligned}
\phi^{(u)}(v(\mathcal{C})) =
\sum_{\mathclap{\mathcal{C'} \subseteq \mathcal{C} \setminus \{O^{(u)}\}}}
\frac{\|\mathcal{C'}\|!(\|\mathcal{C}\| - \|\mathcal{C'}\| - 1)!}{\|\mathcal{C}\|!}\left(v\left(\mathcal{C'} \cup \{O^{(u)}\}\right) - v\left(\mathcal{C'}\right)\right)
\end{aligned}
\end{equation}

Let $\orders_{\mathcal{C}}$ denote all orderings of the organizations from the coalition $\mathcal{C}$.  Each ordering $\prec_{\mathcal{C}}$ can be associated with a permutation of the set $\mathcal{C}$, thus $\|\orders_{\mathcal{C}}\| = \|\mathcal{C}\|!$. For the ordering $\prec_{\mathcal{C}} \in \mathcal{L}_C$ we define $\prec_{\mathcal{C}}(O^{(i)}) = \{O^{(j)} \in \mathcal{C}: O^{(j)} \prec_{\mathcal{C}} O^{(i)}\}$ as the set of all organizations from $\mathcal{C}$ that precede $O^{(i)}$ in the order $\prec_{\mathcal{C}}$. The Shapley value can be alternatively expressed~\cite{RePEc:mtp:titles:0262650401} in the following form: 
\begin{equation}\label{eq::shapleyValueAlt}
\begin{aligned}
\phi^{(u)}(v(\mathcal{C})) =
\frac{1}{\|\mathcal{C}\|!}\sum_{\prec_{\mathcal{C}} \in \orders_{\mathcal{C}}} \left(v\left(\prec_{\mathcal{C}}(O^{(u)}) \cup \{O^{(u)}\}\right) - v\left(\prec_{\mathcal{C}}(O^{(u)}\right)\right) \textrm{.}
\end{aligned}
\end{equation}
\noindent This formulation has an interesting interpretation. Consider the organizations joining the coalition $\mathcal{C}$ in the order $\prec_{\mathcal{C}}$. Each organization $O^{(u)}$, when joining, contributes to the current coalition the value equal to $\left(v(\prec_{\mathcal{C}}(O^{(u)}) \cup \{O^{(u)}\}) - v(\prec_{\mathcal{C}}(O^{(u)})\right)$. Thus, $\phi^{(u)}(v(\mathcal{C}))$ is the expected contribution to the coalition $\mathcal{C}$, when the expectation is taken over the order in which the organizations join  $\mathcal{C}$. Hereinafter we will call the value $\phi^{(u)}(v(\mathcal{C})$ (or using a shorter notation $\phi^{(u)}$) as the \emph{contribution} of the organization $O^{(u)}$.

Let us consider a specific scheduling algorithm $\mathcal{A}$, a specific time moment $t$, and a specific coalition $\mathcal{C}$.
Ideally, the utilities of the organizations should be equal to the reference fair values, 
%\begin{align*}
$\forall_{u}\; \psi^{(u)}(\mathcal{C}) = \phi^{(u)}(v(\mathcal{C}))$,
%\end{align*}
(meaning that the utility of the organization is equal to its contribution), but our scheduling problem is discrete so an algorithm guaranteeing this property may not exist. Thus, we will call as fair an algorithm that results in utilities close to contributions.
The following definition of a fair algorithm is in two ways recursive. A fair algorithm for a coalition $\mathcal{C}$ and time $t$ must be also fair for all subcoalitions $\mathcal{C'} \subset \mathcal{C}$ and for all previous $t' < t$ (an alternative to being fair for all previous $t' < t$ would be to ensure asymptotic fairness; however, our formulation is more responsive and more relevant for the online case. We want to avoid the case in which an organization is disfavored in one, possibly long, time period and then favored in the next one).

\begin{definition}\label{def::fairScheduling}
Set an arbitrary metric $\|\cdot\|_{d}: 2^{k} \times 2^{k} \rightarrow \reals_{\geq 0}$; and set an arbitrary time moment $t \in \timemoments$. $\mathcal{A}$ is a fair algorithm in $t$ for coalition $\mathcal{C}$ in metric $\|\cdot\|_{d}$ if and only if:
\begin{align*}
\mathcal{A} \in \mathrm{argmin}_{\mathcal{A'} \in \mathcal{F}(<t)}\|\vec{\phi}(\mathcal{A'}, \mathcal{C}, t) - \vec{\psi}(v(\mathcal{A'}, \mathcal{C}, t)\|_{d}
\end{align*}
where:
\begin{enumerate}
\item $\mathcal{F}(<t)$ is a set of algorithms fair in each point $t' < t$; $\mathcal{F}(<0)$ is a set of all greedy algorithms,
\item $\vec{\psi}(v(\mathcal{A'}, \mathcal{C})$ is a vector of utilities $\langle \psi^{(u)}(v(\mathcal{A'}, \mathcal{C}))\rangle$,
\item $\vec{\phi}(\mathcal{A'}, \mathcal{C})$ is a vector of contributions $\langle \phi^{(u)}(v(\mathcal{A'}, \mathcal{C})) \rangle$, where $\phi^{(u)}(v(\mathcal{A'}, \mathcal{C}))$ is given by Equation~\ref{eq::shapleyValueAlt},
\item In Equation~\ref{eq::shapleyValueAlt}, for any $\mathcal{C'} \subset \mathcal{C}$, $v(\mathcal{C'})$ denotes $v(\mathcal{A}_{f}, \mathcal{C'})$, where $\mathcal{A}_{f}$ is any fair algorithm for coalition $\mathcal{C'}$.
\end{enumerate}
\end{definition}

\begin{definition}\label{def::fairAlgorithm}
$\mathcal{A}$ is a fair algorithm for coalition $\mathcal{C}$ if and only if it is fair in each time $t \in \timemoments$.
\end{definition}

%In the definitions we require the algorithms to be fair in all time moments $t$. 

Further on, we consider algorithms fair in the Manhattan metric\footnote{Our analysis can be generalized to other distance functions.}: 
%\begin{align*}
$\|\vec{v_1}, \vec{v_2}\|_{M} = \sum_{i=1}^{k}|v_{1}[i] - v_{2}[i]| \textrm{.}$
%\end{align*}

\SetAlCapFnt{\footnotesize}
\begin{figure}[th!]
\begin{algorithm}[H]
   \footnotesize
   \SetKwInput{KwNotation}{Notation}
   \SetKwFunction{Distance}{Distance}
   \SetKwFunction{ReleaseJob}{ReleaseJob}
   \SetKwFunction{FairAlgorithm}{FairAlgorithm}
   \SetKwFunction{UpdateVals}{UpdateVals}
   \SetKwFunction{PropoagateVals}{PropoagateVals}
   \SetKwFunction{FreeMachine}{FreeMachine}
   \SetKwFunction{FreeMachineA}{\textbf{FreeMachine}}
   \SetKwFunction{SelectAndSchedule}{SelectAndSchedule}
   \SetKwBlock{Block}
   \SetAlCapFnt{\footnotesize}
   \KwNotation{\\
	\textbf{jobs}$[\mathcal{C}][O^{(u)}]$ --- list of waiting jobs of organization $O^{(u)}$. \\
	$\pmb{\phi}$$[\mathcal{C}][O^{(u)}]$ --- the contribution of $O^{(u)}$ in $\mathcal{C}$, $\phi^{(u)}(\mathcal{C})$.\\
	$\pmb{\psi}$$[\mathcal{C}][O^{(u)}]$ --- utility of $O^{(u)}$ from being in $\mathcal{C}$, $\psi(\mathcal{C}, O^{(u)})$.\\
	\textbf{v}$[\mathcal{C}]$ --- value of a coalition $\mathcal{C}$.\\
	$\pmb{\sigma}$$[\mathcal{C}]$ --- schedule for a coalition $\mathcal{C}$.\\
    \textbf{\FreeMachineA}($\sigma, t$) --- returns true if and only if there is a free machine in $\sigma$ in time $t$. 
     }
	
        \vspace{0.3cm}\ReleaseJob{$O^{(u)}$, $J$}:
	\Block{
		\For{$\mathcal{C}: O^{(u)} \in \mathcal{C}$}
		{
			$\mathrm{jobs}[\mathcal{C}][O^{(u)}].\mathrm{push}(J)$
		}
	}

	\vspace{0.3cm}\Distance{$\mathcal{C}$, $O^{(u)}$, $t$}:
	\Block{
		$old \leftarrow \sigma[\mathcal{C}]$\;
		$new \leftarrow \sigma[\mathcal{C}] \cup \{(\mathrm{jobs}[\mathcal{C}][O^{(u)}].\mathrm{first}, t)\}$\;
		$\Delta \psi \leftarrow \psi(new, O^{(u)}, t) - \psi(old, O^{(u)}, t)$\;
		\textbf{return} $\left|\phi[\mathcal{C}][O^{(u)}] + \frac{\Delta\psi}{\|\mathcal{C}\|} - \psi[\mathcal{C}][O^{(u)}] - \Delta\psi\right|$ \\
		$+\sum_{O^{(u')}}\left|\phi[\mathcal{C}][O^{(u')}] + \frac{\Delta\psi}{\|\mathcal{C}\|} - \psi[\mathcal{C}][O^{(u')}]\right|$\;
	}

	\vspace{0.3cm}\SelectAndSchedule{$\mathcal{C}$, $t$}:
	\Block{
		$u \leftarrow \mathrm{argmin}_{O^{(u)}}(\Distance(\mathcal{C}, O^{(u)}, t))$ \;
		$\sigma[\mathcal{C}] \leftarrow \sigma[\mathcal{C}] \cup \{(\mathrm{jobs}[\mathcal{C}][u].\mathrm{first}, t)\}$\;
		$\psi[\mathcal{C}][O^{(u)}] \leftarrow \psi(\sigma[\mathcal{C}], O^{(u)}, t)$\;
	}

	\vspace{0.3cm}\UpdateVals{$\mathcal{C}$, $t$}:
	\Block{
		\ForEach{$O^{(u)} \in \mathcal{C}$}
		{
			$\psi[\mathcal{C}][O^{(u)}] \leftarrow \psi(\sigma[\mathcal{C}], O^{(u)}, t)$\;
			$\phi[\mathcal{C}][O^{(u)}] \leftarrow 0$\;
		}
		$\mathrm{v}[\mathcal{C}] \leftarrow \sum_{ O^{(u)}}\psi(\sigma[\mathcal{C}], O^{(u)}, t) $\;

		\ForEach{$\mathcal{C}_{sub}$: $\mathcal{C}_{sub} \subseteq \mathcal{C}$}
		{ \nllabel{algline::line1}
			\ForEach{$O^{(u)} \in \mathcal{C}_{sub}$}
			{
				 $\phi[\mathcal{C}][O^{(u)}] \leftarrow \phi[\mathcal{C}][O^{(u)}] +$ \\
					\hspace{5mm} $(\mathrm{v}[\mathcal{C}_{sub}] - \mathrm{v}[\mathcal{C}_{sub} \setminus \{O^{(u)}\}])$ \\
					\hspace{5mm} $\cdot \frac{(\|\mathcal{C}_{sub}\| - 1)!(\|\mathcal{C}\| - \|\mathcal{C}_{sub}\|)!}{\|\mathcal{C}\|!}$ \nllabel{algline::line2} \;
			}
		}
	}

	\vspace{0.3cm}\FairAlgorithm{$\mathcal{C}$}:
	\Block{
		\ForEach{time moment $t$}
		{
			\ForEach{job $J_i^{(u)}$: $r_i^{(u)} = t$}
			{
				\ReleaseJob{$O_i^{(u)}, J_i^{(u)}$}\;
			}
			\For{$s \leftarrow 1$ \KwTo $\|C\|$}
			{
				\ForEach{$\mathcal{C'} \subset \mathcal{C}$, such that $\|C'\| = s$}
				{
					\UpdateVals{$\mathcal{C'}, t$}\;
					\While{\FreeMachine{$\sigma[\mathcal{C'}]$, $t$}}
					{
						\SelectAndSchedule{$\mathcal{C'}, t$}\;
					}
					$\mathrm{v}[\mathcal{C}] \leftarrow \sum_{ O^{(u)}}\psi(\sigma[\mathcal{C}], O^{(u)}, t) $\;
				}
			}
		}
	}

\end{algorithm}
\caption{Algorithm~\textsc{Ref}: a fair algorithm for arbitrary utility function $\psi$.}\label{alg:fairAlg}
\end{figure}

Based on Definition~\ref{def::fairAlgorithm} we construct a referral fair algorithm for an arbitrary utility function $\psi$ (Algorithm~\textsc{Ref}; the pseudo-code is presented in Figure~\ref{alg:fairAlg}). Algorithm~\textsc{Ref} keeps a schedule for every subcoalition $\mathcal{C'} \subset \mathcal{C}$. 
For each time moment the algorithm complements the schedule starting from the subcoalitions of the smallest size. The values of all smaller coalitions $v[\mathcal{C}_s]$ are used to update the contributions of the organizations (lines~\ref{algline::line1}-\ref{algline::line2}) in the procedure \texttt{UpdateVals}). 
Before scheduling any job of the coalition $\mathcal{C}'$ the contribution and the utility of each organization in $\mathcal{C}'$ is updated (procedure \texttt{UpdateVals}). 
If there is a free machine and a set of jobs waiting for execution, the algorithm selects the job according to Definition~\ref{def::fairScheduling}, thus it selects the organization that minimizes the distance of the utilities $\vec{\psi}$ to their ideal values $\vec{\phi}$ (procedure \texttt{SelectAndSchedule}). 
Assuming the first job of the organization $O^{(u)}$ is tentatively scheduled, the procedure \texttt{Distance} computes a distance between the new values of $\vec{\psi}$ and $\vec{\phi}$.
 
The procedure \texttt{Distance} works as follows. 
Assuming $O^{(u)}$ is selected the value $\Delta\psi$ denotes the increase of the utility of $O^{(u)}$ thanks to scheduling its first waiting job. This is also the increase of the value of the whole coalition. When procedure \texttt{Distance}($\mathcal{C}, O^{(u)}, t$) is executed, the schedules (and thus, the values) in time $t$ for all subcoalitions $\mathcal{C}' \subset \mathcal{C}$ are known. The schedule, for coalition $\mathcal{C}$ is known only in time $(t-1)$, as we have not yet decided which job should be scheduled in $t$. Thus, scheduling the job will change the schedule (and the value) only for a coalition $\mathcal{C}$. From the definition of the Shapley value it follows that if the value $v(\mathcal{C})$ of the coalition $\mathcal{C}$ increases by $\Delta\psi$ and the value of all subcoalitions remains the same, then the contribution $\phi^{(u')}$ of each organization $O^{(u')} \in \mathcal{C}$ to $\mathcal{C}$ will increase by the same value equal to $\Delta\psi/\|\mathcal{C}\|$. Thus, for each organization $O^{(u')} \in \mathcal{C}$ the new contribution of $O^{(u')}$ is $(\phi[\mathcal{C}][O^{(u')}] + \frac{\Delta\psi}{\|\mathcal{C}\|})$. The new utility for each organization $O^{(u')} \in \mathcal{C}$, such that $O^{(u')} \neq O^{(u)}$ is equal to $\psi[\mathcal{C}][O^{(u')}]$. The new utility of the organization $O^{(u)}$ is equal to $(\psi[\mathcal{C}][O^{(u)}]| + \Delta\psi)$.

% From Definition~\ref{def::fairAlgorithm} we get the following theorem:

\begin{theorem}
Algorithm~\textsc{Ref} from Figure~\ref{alg:fairAlg} is a fair algorithm.
\end{theorem}\label{thm::fairnessInGeneral}
\begin{proof}
Algorithm~\textsc{Ref} is a straightforward implementation of Definition~\ref{def::fairAlgorithm}.
\end{proof}

\begin{proposition}
In each time moment $t$ the time complexity of Algorithm~\textsc{Ref} from Figure~\ref{alg:fairAlg} is \\ $O(\|\orgs\|(2^{\|\orgs\|}\sum m^{(u)} + 3^{\|\orgs\|}))$.
\end{proposition}
\begin{proof}
Once the contribution is calculated,  each coalition in $t$ may schedule at most $\sum m^{(u)}$ jobs. The time needed for selecting each such a job is proportional to the number of the organizations. Thus, we get the $\|\orgs\|2^{\|\orgs\|}\sum m^{(u)}$ part of the complexity. For calculating the contribution of the organization $O^{(u)}$ to the coalition $\mathcal{C}$ the algorithm considers all subsets of $\mathcal{C}$ -- there are $2^{\|\mathcal{C}\|}$ such subsets. Since there are $\|\orgs\| \choose k$ coalitions of size $k$, the number of the operations required for calculating the contributions of all organizations is proportional to:
\begin{align*}
\sum_{(u)} \sum_{k = 0}^{\|\orgs\|}{\|\orgs\| \choose k} 2^{k} = \|\orgs\|\sum_{k = 0}^{\|\orgs\|}{\|\orgs\| \choose k} 1^{\|\orgs\| - k} 2^{k} = \|\orgs\|(1 + 2)^{\|\orgs\|} = \|\orgs\|3^{\|\orgs\|} \textrm{.}
\end{align*}
This gives the $\|\orgs\|3^{\|\orgs\|}$ part of the complexity and completes the proof.
\end{proof}

\begin{corollary}
The problem of finding fair schedule parameterized with the number of organizations is FPT.
\end{corollary}

\section{Strategy-proof utility functions}\label{sec::strategyProof}

There are many utility functions considered in scheduling, e.g. flow time, turnaround time, resource utilization, makespan, tardiness. However, it is not sufficient to design a fair algorithm for an arbitrary utility function $\psi$. Some functions may create incentive for organizations to manipulate their workload: to divide the tasks into smaller pieces, to merge or to delay them. This is undesired as an organization should not profit nor suffer from the way it presents its workload. An organization should present their jobs in the most convenient way; it should not focus on playing against other organizations. We show that in organizationally distributed systems, as we have to take into account such manipulations, the choice of the utility functions is restricted. 

For the sake of this section we introduce additional notation: let us fix an arbitrary organization $O^{(u)}$ and let $\sigma_t$ denote a schedule of the jobs of $O^{(u)}$ in time $t$. The jobs $J_i(s_i, p_i)$ of $O^{(u)}$ are characterized by their start times $s_i$ and processing times $p_i$. We are considering envy-free utility functions that for a given organization $O^{(u)}$ depend only on the schedule of the jobs of $O^{(u)}$. This means that there is no external economical relation between the organization (the organization $O^{u}$ cares about $O^{v}$ only if the jobs of $O^{v}$ influence the jobs of $O^{u}$ -- in contrast to looking directly at the utility of $O^{v}$).
We also assume the non-clairvoyant model -- the utility in time $t$ depends only on the jobs or the parts of the jobs completed before or at $t$. Let us assume that our goal is to maximize the utility function\footnote{We can easily transform the problem to the minimization form by taking the inverse of the standard maximization utility function}. We start from presenting the desired properties of the utility function~$\psi$ (when presenting the properties we use the shorter notation $\psi(\sigma_t)$ for $\psi(\sigma_t, t)$):

\begin{enumerate}

\item Tasks anonymity (starting times) --- improving the completion time of a single task with a certain processing time $p$ by one unit of time is for each task equally profitable -- for $s, s' \leq t-1$, we require:
\begin{align*}
	\psi(\sigma_t \cup \{(s, p)\}) - \psi(\sigma_t \cup \{(s + 1, p)\}) =
	\psi(\sigma_t' \cup \{(s', p)\}) - \psi(\sigma_t' \cup \{(s' + 1, p)\}) > 0\textrm{.}
\end{align*}

\item Tasks anonymity (number of tasks) --- in each schedule increasing the number of completed tasks is equally profitable -- for $s \leq t-1$, we require:
\begin{align*}
	\psi(\sigma_t \cup \{(s, p)\}) - \psi(\sigma_t) =
	\psi(\sigma_t' \cup \{(s, p)\}) - \psi(\sigma_t') > 0 \textrm{.}
\end{align*}

\item Strategy-resistance --- the organization cannot profit from merging multiple smaller jobs into one larger job or from dividing a larger job into smaller pieces:
\begin{align*}
 	\psi(\sigma_t \cup \{(s, p_1)\}) + \psi(\sigma_t \cup \{(s + p_1, p_2)\}) =
	\psi(\sigma_t \cup \{(s, p_1 + p_2)\}) \textrm{.}
\end{align*}
In spite of dividing and merging the jobs, each organization can delay the release time of their jobs and artificially increase the size of the jobs. Delaying the jobs is however never profitable for the organization (by property 1). Also, the strategy-resistance property discourages the organizations to increase the sizes of their jobs (the utility coming from processing a larger job is always greater).
\end{enumerate}

To within a multiplicative and additive constants, there is only one utility function satisfying the aforementioned properties.

\begin{theorem}\label{thm:strategyProofUtility}
Let $\psi$ be a utility function that satisfies the 3 properties: task anonymity (starting times); task anonymity (number of tasks); strategy-resistance. $\psi$ is of the following form:
\begin{align*}
	\psi(\sigma, t) = \sum_{\mathclap{(s, p) \in \sigma_t}}\min(p, t-s)(K_1 - K_2\frac{s+ \min(s+p-1, t-1)}{2}) + K_3 \textrm{,}
\end{align*}
where
\begin{enumerate}
\item $K_1 = \psi(\sigma \cup \{(0, 1)\}, t) - \psi(\sigma) > 0$
\item $K_2 = \psi(\sigma \cup \{(s, p)\}, t) - \psi(\sigma \cup \{(s+1, p)\}, t) > 0$
\item $K_3 = \psi(\emptyset) \textrm{.}$
\end{enumerate}
\end{theorem}
\begin{proof}
\begin{align*}
\psi(\sigma, t) & = \psi(\bigcup_{\mathclap{(s, p) \in \sigma}}(s, p), t) = \psi(\bigcup_{\mathclap{(s, p) \in \sigma}}(s, \min(p, t-s)), t) \\
        & \textrm{(non-clairvoyance)} \\
	& =\psi(\bigcup_{(s, p) \in \sigma}\bigcup_{i = s}^{\min(s+p-1, t-1)}(i, 1), t) \\
        & \textrm{(strategy-resistance)}\\
	& = \psi(\bigcup_{(s, p) \in \sigma}\bigcup_{i = s}^{\min(s+p-1, t-1)}(0, 1), t) - K_2\sum_{(s, p) \in \sigma_t}\sum_{i=s}^{\mathclap{\min(s+p-1, t-1)}} i \\
        & \textrm{(starting times anonymity)}\\
        & = \psi(\emptyset) + \sum_{\mathclap{(s, p) \in \sigma_t}}\sum_{i = s}^{\min(s+p-1, t-1)}K_1 \\
        & \textrm{(number of tasks anonymity)}\\
        & \;\;\;\; - K_2\sum_{(s, p) \in \sigma_t}\min(p, t-s)\frac{s+ \min(s+p-1, t-1)}{2} \\
        & \textrm{(sum of the arithmetic progression)}\\
	& = K_3 + \sum_{\mathclap{(s, p) \in \sigma_t}}\min(p, t-s)(K_1 - K_2\frac{s+ \min(s+p-1, t-1)}{2}) &
\end{align*}
\end{proof}

We set the constants $K_1, K_2, K_3$ so that to simplify the form of the utility function and ensure that the utility is always positive. With $K_1 = 1$, $K_2 = t$ and $K_3 = 0$, we get the following strategy-proof utility function: 
\begin{align}\label{def::fairMetric}
\psi_{sp}(\sigma, t) = \sum_{\mathclap{(s, p) \in \sigma: s \leq t}}\min(p, t-s)\left(t - \frac{s+ \min(s+p-1, t-1)}{2}\right) \textrm{.}
\end{align}

$\psi_{sp}$ can be interpreted as the task throughput. A task with processing time $p_i$ can be identified with $p_i$ unit-sized tasks starting in consecutive time moments. Intuitively, the function $\psi_{sp}$ assigns to each such unit-sized task starting at time $t_s$ a utility value equal to $(t - t_s)$; the higher the utility value, the earlier this unit-sized task completes. A utility of the schedule is the sum of the utilities over all such unit-sized tasks.
$\psi_{sp}$ is similar to the flow time 
% (which is the completion time minus the release time summed over all jobs) 
except for two differences: (i) Flow time is a minimization objective, but increasing the number of completed jobs increases its value. E.g., scheduling no jobs results in zero (optimal) flow time, but of course an empty schedule cannot be considered optimal (breaking the second axiom);
 % -- this is an undesired behavior that breaks the second axiom -- the tasks anonymity with respect to the number of the tasks. 
(ii) Flow time favors short tasks, which is an incentive for dividing tasks into smaller pieces (this breaks strategy-resistance axiom). The differences between the flow time and $\psi_{sp}$ is also presented on example in Figure~\ref{fig:fairUtilityExample}. The similarity of $\psi_{sp}$ to the flow time is quantified by Proposition~\ref{thm::metricLikeFlowTime} below.

\begin{proposition}\label{thm::metricLikeFlowTime}
Let $\mathcal{J}$ be a fixed set of jobs, each having the same processing time $p$ and each completed before $t$. Then, maximization of the $\psi_{sp}$ utility is equivalent to minimization of the flow time of the jobs.
\end{proposition}
\begin{proof}
Let $\sigma$ denote an arbitrary schedule of $\mathcal{J}$. Since the flow time uses the release times of the jobs, we will identify the jobs with the triples $(s, p, r)$ where $s$, $p$ and $r$ denote the start time, processing time and release time, respectively. Let $\psi_{ft}(\sigma)$ denote the total flow time of the jobs from $\mathcal{J}$ in schedule $\sigma$. We have:
\begin{align*}
\psi_{sp}(\sigma, t) &= \sum_{\mathclap{(s, p, r) \in \sigma: s \leq t}}\min(p, t-s)\left(t - \frac{s+ \min(s+p-1, t-1)}{2}\right) &\\
        &= \sum_{\mathclap{(s, p, r) \in \sigma}}p\left(t - \frac{2s + p-1}{2}\right) \\
        & \textrm{(each job is completed before $t$)} \\
	&= \sum_{\mathclap{(s, p, r) \in \sigma}}(pt + \frac{p^2 + p}{2} - r) - p\sum_{\mathclap{(s, p, r) \in \sigma}}((p+s) - r) & \\
	&= \|\mathcal{J}\|(pt + \frac{p^2 + p}{2}) - \sum_{\mathclap{(s, p, r) \in \sigma}}(r) - p\psi_{ft}(\sigma) &
\end{align*}
Since $p$, $\|\mathcal{J}\|(pt + \frac{p^2 + p}{2})$ and $\sum_{(s, p, r) \in \sigma}r$ are constants we get the thesis.
\end{proof}

\begin{figure}[tb]
  \begin{center}
    \includegraphics[scale=1.25]{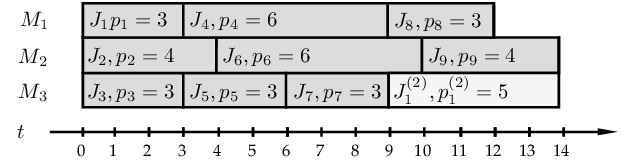}
  \end{center}
  %\vspace{-0.75cm}
  \caption{Consider 9 jobs owned by $O^{(1)}$ and a single job owned by $O^{(2)}$, all scheduled on 3 processors. We assume all jobs were released in time 0. In this example all jobs finish before or at time $t = 14$. The utility $\psi_{sp}$ of the organization $O^{(1)}$ in time 13 does not take into account the last uncompleted unit of the job $J_9$, thus it is equal to: $3\cdot(13 - \frac{0+2}{2}) + 4\cdot(13 - \frac{0+3}{2}) + \dots + 3\cdot(13 - \frac{9+11}{2}) + 3\cdot(13 - \frac{10+12}{2}) = 262$. The utility in time 14 takes into account all the parts of the jobs, thus it is equal to  $3\cdot(14 - \frac{0+2}{2}) + 4\cdot(14 - \frac{0+3}{2}) + \dots + 3\cdot(14 - \frac{9+11}{2}) + 4\cdot(14 - \frac{10+13}{2}) = 297$.  The flow time in time 14 is equal to $3 + 4+ \dots + 14 = 70$. If there was no job $J_1^{(2)}$, then $J_9$ would be started in time 9 instead of 10 and the utility $\psi_{sp}$ in time 14 would increase by $4\cdot(\frac{10 + 13}{2} - \frac{9 + 12}{2}) = 4$ (the flow time would decrease by 1). If, for instance, $J_6$ was started one time unit later, then the utility of the schedule would decrease by 6 (the flow time would decrease by 1), which shows that the utility takes into account the sizes of the jobs (in contrast to the flow time). If the job $J_9$ was not scheduled at all, 
the utility $\psi_{sp}$ would decrease by 10, which shows that the schedule with more tasks has higher (more optimal) utility (the flow time would decrease by 14; since flow time is a minimization metric, this breaks the second axiom regarding the tasks anonymity).}
  \label{fig:fairUtilityExample}
\end{figure}

\section{Fair scheduling with strategy-proof utility}

\begin{figure}[th]
\begin{algorithm}[H]
   \footnotesize
   \SetKwInput{KwNotation}{Notation}
   \SetKwFunction{Distance}{Distance}
   \SetKwFunction{ReleaseJob}{ReleaseJob}
   \SetKwFunction{FairAlgorithm}{FairAlgorithm}
   \SetKwFunction{UpdateVals}{UpdateVals}
   \SetKwFunction{PropoagateVals}{PropoagateVals}
   \SetKwFunction{FreeMachine}{FreeMachine}
   \SetKwFunction{SelectAndSchedule}{SelectAndSchedule}
   \SetKwBlock{Block}
	\SelectAndSchedule{$\mathcal{C}$, $t$}:
	\Block{
		$u \leftarrow \mathrm{argmin}_{O^{(u)}}(\psi[\mathcal{C}][O^{(u)}] - \phi[\mathcal{C}][O^{(u)}])$ \;
		$\sigma[\mathcal{C}] \leftarrow \sigma[\mathcal{C}] \cup \{(\mathrm{jobs}[\mathcal{C}][u].\mathrm{first}, t)\}$\;
		$\psi[\mathcal{C}][O^{(u)}] \leftarrow \psi(\sigma[\mathcal{C}], O^{(u)}, t)$\;
	}
\end{algorithm}
\caption{\footnotesize Function \texttt{SelectAndSchedule} for utility function $\psi_{sp}$.}\label{alg:fairAlgConc}
\end{figure}

For the concrete utility function $\psi_{sp}$ we can simplify the \texttt{SelectAndSchedule} function in Algorithm~\textsc{Ref}. The simplified version is presented in Figure~\ref{alg:fairAlgConc}.

The algorithm selects the organization $O^{(u)}$ that has the largest difference ($\phi^{(u)} - \psi^{(u)}$) that is the organization that has the largest contribution in comparison to the obtained utility. One can wonder whether we can select the organization in polynomial time -- without keeping the $2^{\|\mathcal{C}\|}$ schedules for all subcoalitions. Unfortunately, the problem of calculating the credits for a given organization is NP-hard. 

\begin{theorem}\label{thm::contributionsNp}
The problem of calculating the contribution $\phi^{(u)}(\mathcal{C}, t)$ for a given organization $O^{(u)}$ in coalition $\mathcal{C}$ in time $t$ is NP-hard.
\end{theorem}
\begin{proof}
We present the reduction of the $\textsc{SubsetSum}$ problem (which is NP-hard) to the problem of calculating the contribution for an organization. Let $I$ be an instance of the $\textsc{SubsetSum}$ problem. In $I$ we are given a set of $k$ integers $S = \{x_1, x_2, \dots, x_k\}$ and a value $x$. We ask whether there exists a subset of $S$ with the sum of elements equal to $x$. From $I$ we construct an instance $I_{con}$ of the problem of calculating the contribution for a given organization. Intuitively, we construct the set of $(\|S\| + 2)$ organizations: $\|S\|$ of them will correspond to the appropriate elements from $S$. The two dummy organizations $a$ and $b$ are used for our reduction. One dummy organization $a$ has no jobs. The second dummy organization $b$ has a large job that dominates the value of the whole schedule. The instance $I_{con}$ is constructed in such a way that for each coalition $\mathcal{C}$ such that $b \in \mathcal{C}$ and such that the elements of $S$ corresponding to the organizations from $\mathcal{C}$ sum up to the value lower than $x$, the marginal contribution of $a$ to 
$\mathcal{C}$ is $L + O(L)$, where $O(L)$ is small in comparison with $L$. The marginal contribution of $a$ to other coalitions is small ($O(L)$). Thus, from the contribution of $a$, we can count the subsets of $S$ with the sum of the elements lower than $x$. By repeating this procedure for $(x+1)$ we can count the subsets of $S$ with the sum of the elements lower than $(x+1)$. By comparing the two values, we can find whether there exists the subset of $S$ with the sum of the elements equal to $x$. 
The precise construction is described below.

Let $\mathcal{S}_{<x} = \{S' \subset S: \sum_{x_i \in S'}s_i < x\}$ be the set of the subsets of $S$, each having the sum of the elements lower than $x$.
Let $n_{< x}(S) = \sum_{S' \in \mathcal{S}_{<x}} (\|S'\| + 1)!(\|S\| - \|S'\|)!$ be the number of the orderings (permutations) of the set $S \cup \{a, b\}$ that starts with some permutation of the sum of exactly one element of $\mathcal{S}_{<x}$ (which is some subset of $S$ such that the sum of the elements of this subset is lower than $x$) and $\{b\}$ followed by the element $a$. In other words, if we associate the elements from $S \cup \{a, b\}$ with the organizations and each ordering of the elements of $S \cup \{a, b\}$  with the order of the organizations joining the grand coalition, then $n_{< x}(S)$ is the number of the orderings corresponding to the cases when organization $a$ joins grand coalition just after all the organizations from $S' \cup \{b\}$, where $S'$ is some element of $\mathcal{S}_{<x}$.
Of course $\mathcal{S}_{<x} \subseteq \mathcal{S}_{<(x+1)}$.
Note that there exists $S' \subset S$, such that $\sum_{x_i \in S'}x_i = x$ if and only if the set $\mathcal{S}_{<x}$ is a proper subset of $\mathcal{S}_{<(x+1)}$ (i.e. $\mathcal{S}_{<x} \subset \mathcal{S}_{<(x+1)}$). 
Indeed, there exists $S'$ such that $S' \notin \mathcal{S}_{<x}$ and $S' \in \mathcal{S}_{<(x+1)}$ if and only if $\sum_{x_i \in S'}x_i < x+1$ and $\sum_{x_i \in S'}x_i \geq x$ from which it follows that $\sum_{x_i \in S'}x_i = x$.
Also, $\mathcal{S}_{<x} \subset \mathcal{S}_{<(x+1)}$ if and only if $n_{< (x+1)}(S)$ is greater than $n_{< (x)}(S)$ (we are doing a summation of the positive values over the larger set).

In $I_{con}$ there is a set of $(k+2)$ machines, each owned by a different organization. We will denote the set of first $k$ organizations as $\mathcal{O}_{S}$, the (k+1)-th organization as $a$ and the  (k+2)-th organization as $b$. Let $x_{tot} = \sum_{j = 1}^{k}x_{j} + 2$. The $i$-th organization from $\mathcal{O}_{S}$ has 4 jobs: $J^{(i)}_1, J^{(i)}_2, J^{(i)}_3$ and $J^{(i)}_4$, with release times $r^{(i)}_1 = r^{(i)}_1 = 0$, $r^{(i)}_3 = 3$ and $r^{(i)}_4 = 4$; and processing times $p^{(i)}_1 = p^{(i)}_2 = 1$, $p^{(i)}_3 = 2x_{tot}$ and $p^{(i)}_4 = 2x_i$.
The organization $a$ has no jobs; the organization $b$ has two jobs $J^{(b)}_1$ and $J^{(b)}_2$, with release times $r^{(b)}_1 = 2$ and $r^{(b)}_2 = (2x+3)$; and processing times $p^{(b)}_1 = (2x+2)$ and $p^{(b)}_2 = L = 4\|S\|x_{tot}^2((k+2)!) + 1$ (intuitively $L$ is a large number).

Until time $t = 2$ only the organizations from $\mathcal{O}_{S}$ have some (unit-size) jobs to be executed. The organization $b$ has no jobs till time $t = 2$, so it will run one or two unit-size jobs of the other organizations, contributing to all such coalitions that include $b$ and some other organizations from $\mathcal{O}_{S}$. This construction allows to enforce that in the first time moment after $t=2$ when there are jobs of some of the organizations from $\mathcal{O}_{S}$ and of $b$ available for execution, the job of $b$ will be selected and scheduled first.

\begin{figure}[tb]
  \begin{center}
    \includegraphics[scale=0.8]{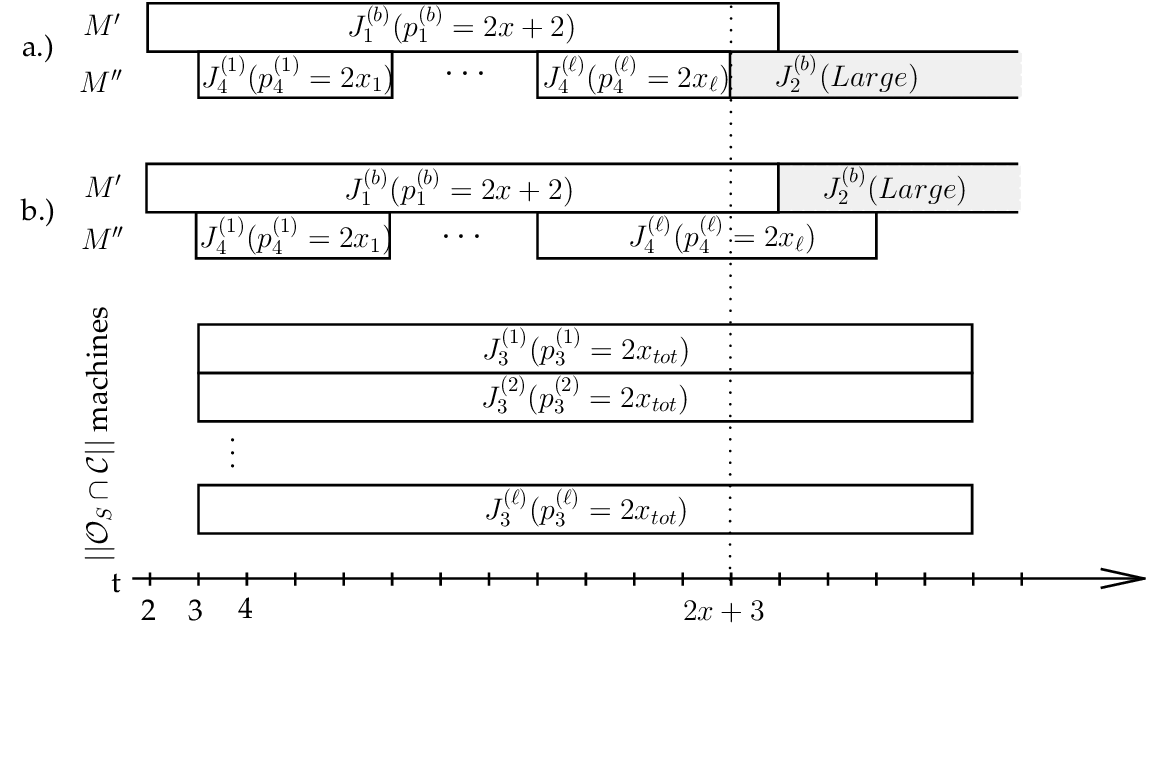}
  \end{center}
  \vspace{-2.2cm}
  \caption{The schedules for the coalition $\mathcal{C} \cup \{a\}$ for two cases: a)~$\sum_{i: O^{(i)} \in \mathcal{C} \land O^{(i)} \in \mathcal{O}_{S}}x_i  \leq x$, b)~$\sum_{i: O^{(i)} \in \mathcal{C} \land O^{(i)} \in \mathcal{O}_{S}}x_i  > x$. The two cases a) and b) differ only in the schedules on machines $M'$ and $M''$. In the case a) the large job $J_2^{(b)}$ (marked as a light gray) is started one time unit earlier than in case b).}
  \label{fig:npHard}
\end{figure}

Let us consider a contribution of $a$ to the coalition $\mathcal{C}$ such that $a \notin \mathcal{C}$ and $b \in \mathcal{C}$. 
There are $(\|\mathcal{C} \cap \mathcal{O}_{S}\|+2)$ machines in the coalition $\mathcal{C} \cup \{a\}$.
The schedule in $\mathcal{C} \cup \{a\}$ after $t=2$ looks in the following way (this schedule is depicted in Figure~\ref{fig:npHard}). 
In time $t=2$ one machine (let us denote this machine as $M'$) starts the job $J^{(b)}_1$ 
In time $t = 3$ some $\|\mathcal{C} \cap \mathcal{O}_{S}\|$ machines start the third jobs (the one with size $2x_{tot}$) of the organizations from $\mathcal{C} \cap \mathcal{O}$ and one machine (denoted as $M''$) starts the fourth jobs of the organizations from $\mathcal{C} \cap \mathcal{O}_{S}$; the machine $M''$ completes processing all these jobs in time $2y + 4$, where $y = \sum_{i: O^{(i)} \in \mathcal{C} \land O^{(i)} \in \mathcal{O}_{S}}x_i$ (of course $2y + 4\leq 2x_{tot}$).
In time $(2x+3)$, if $y < x$ the machine $M''$ starts processing the large job $J^{(b)}_2$ of the organization $b$; otherwise machine $M''$ in time $(2x+3)$ still executes some job $J^{(i)}_4$ (as the jobs $J^{(i)}_4$ processed on $M''$ start in even time moments).
In time $2x+4$, if $y \geq x$, the large job $J^{(b)}_2$ is started by machine $M'$  just after the job $J^{(b)}_1$ is completed, ($J^{(b)}_1$ completes in $(2x+4)$); here we use the fact that after $t=2$, $b$ will be prioritized over the organizations from $\mathcal{O}_{S}$. To sum up: if $y < x$ then the large job $J^{(b)}_2$ is started in time $(2x+3)$, otherwise it is started in time $(2x+4)$.

If $y < x$ then by considering only a decrease of the starting time of the largest job, the contribution of $a$ to the coalition $\mathcal{C}$ can be lower bounded by $c_1$:
\begin{align*}
c_1 = L\left(t - \frac{(2x+3)+(2x+3+L)}{2}\right) - L\left(t - \frac{(2x+4) + (2x+4+L))}{2}\right) = L \mathrm{,}
\end{align*}
The organization $a$ causes also a decrease of the starting times of the small jobs (the jobs of the organizations from $\mathcal{O}_{S}$); each job of size smaller or equal to $2x_{tot}$. The starting time of each such small job is decreased by at most $2x_{tot}$ time units. Thus, the contribution of $a$ in case $y < x$ can be upper bounded by $c_2$:
\begin{align*}
c_2 \leq L + 4\|S\|x_{tot}^2 \mathrm{.}
\end{align*}

If $y \geq x$ then $a$ causes only a decrease of the starting times of the small jobs of the organizations from $\mathcal{O}_{S}$, so the contribution of $a$ to $\mathcal{C}$ in this case can be upper bounded by $c_3$:
\begin{align*}
c_3 \leq 4\|S\|x_{tot}^2 \mathrm{.}
\end{align*}
By similar reasoning we can see that the contribution of $a$ to any coalition $\mathcal{C}'$ such that $b \notin \mathcal{C}'$ is also upper bounded by $4\|S\|x_{tot}^2$.

The contribution of organization $a$, $\phi^{(a)}$, is given by Equation~\ref{eq::shapleyValueAlt}, with $u=a$ and $C = \{O^{(1)} \dots O^{(k+2)}\}$. Thus:
\begin{align*}
\phi^{(a)} = \sum_{\mathclap{\mathcal{C'} \subseteq \mathcal{C} \setminus \{a\}}}
\frac{\|\mathcal{C'}\|!(k + 1 - \|\mathcal{C'}\|)!}{(k+2)!}\mathrm{marg}\_\phi(\mathcal{C'}, a) \textrm{,}
\end{align*}
where $\mathrm{marg}\_\phi(\mathcal{C'}, a)$ is the contribution of $a$ to coalition $\mathcal{C'}$. All the coalitions $\mathcal{C'}$ such that $a \notin \mathcal{C'}$, $b \in \mathcal{C'}$ and $\sum_{i: O^{(i)}\in \mathcal{C'} \cap \mathcal{O}_{S}}x_i < x$ will contribute to $\phi^{(a)}$ the value at least equal to $\frac{n_{<x}(S)}{(k+2)!}c_1 = \frac{n_{<x}(S)L}{2(k+2)!}$ (as there is exactly $n_{<x}(S)$ orderings corresponding to the the case when $a$ is joining such coalitions $\mathcal{C'}$) and at most equal to $\frac{n_{<x}(S)}{(k+2)!}c_2 \leq \frac{n_{<x}(S)(L + 8\|S\|x_{tot}^2)}{2(k+2)!}$. The other $(k+2)! - n_{<x}(S)$ orderings will contribute to $\phi^{(a)}$ the value at most equal to $\frac{((k+2)! - n_{<x}(S))}{(k+2)!}c_3 = \frac{((k+2)! - n_{<x}(S))(4\|S\|x_{tot}^2)}{(k+2)!}$. Also:
\begin{align*}
\frac{((k+2)! - n_{<x}(S))(4\|S\|x_{tot}^2)}{(k+2)!} + \frac{n_{<x}(S)(4\|S\|x_{tot}^2)}{(k+2)!} = 4\|S\|x_{tot}^2 < \frac{L}{(k+2)!} \textrm{,}
\end{align*}
which means that $\phi^{(a)}$ can be stated as  $\phi^{(a)} = \frac{n_{<x}(S)L}{(k+2)!} + R$, where $0 \leq R \leq \frac{L}{(k+2)!}$.
We conclude that $\lfloor \frac{(k+2)!\phi^{(a)}}{L} \rfloor =  n_{<x}(S)$. We have shown that calculating the value of $\phi^{(a)}$ allows us to find the value $n_{<x}(S)$. Analogously, we can find $n_{<(x+1)}(S)$. By comparing $n_{<x}(S)$ with $n_{<(x+1)}(S)$ we find the answer to the initial \textsc{SubsetSum} problem, which completes the proof.

\end{proof}

We propose the following definition of the approximation of the fair schedule (similar definitions of the approximation ratio are used for multi-criteria optimization problems~\cite{Ehrgott00approximationalgorithms}):
\begin{definition}
Let $\sigma$ be a schedule and let $\vec{\psi}$ be a vector of the utilities of the organizations in $\sigma$. 
We say that $\sigma$ is an $\alpha$-approximation fair schedule in time $t$ if and only if there exists a truly fair schedule $\sigma^{*}$, with the vector $\vec{\psi}^{*} = \langle \psi^{(u),*} \rangle$ of the utilities of the organizations, such that:
\begin{align*}
\|\vec{\psi} - \vec{\psi}^{*}\|_{M} \leq \alpha \|\vec{\psi}^{*}\|_{M} = \alpha \sum_{u} \psi^{(u),*} = \alpha \cdot v(\sigma^{*}, \mathcal{C}) \textrm{.}
\end{align*}
\end{definition}

Unfortunately, the problem of finding the fair schedule is difficult to approximate. There is no algorithm better than 1/2 (the proof below). This means that the problem is practically inapproximable. Consider two schedules of jobs of $m$ organizations on a single machine. Each organization has one job; all the jobs are identical. In the first schedule $\sigma_{ord}$ the jobs are scheduled in order: $J_{1}^{(1)}, J_{1}^{(2)}, \dots J_{1}^{(m)}$ and in the second schedule $\sigma_{rev}$ the jobs are scheduled in exactly reverse order: $J_{1}^{(m)}, J_{1}^{(m-1)}, \dots J_{1}^{(1)}$. The relative distance between $\sigma_{ord}$ and $\sigma_{rev}$ tends to 1 (with increasing $m$), so $(\frac{1}{2})$-approximation algorithm does not allow to decide whether $\sigma_{ord}$ is truly better than $\sigma_{rev}$. In other words, $\frac{1}{2})$-approximation algorithm cannot distinguish whether a given order of the priorities of the organizations is more fair then the reverse order.

\begin{theorem}\label{thm::inapprox}
For every $\epsilon > 0$, there is no polynomial algorithm for finding the $(\frac{1}{2} - \epsilon)$-approximation fair schedule, unless $\p = \np$.
\end{theorem}
\begin{proof}
Intuitively, we divide time in $(\|\mathcal{B}\|^2 + 3)$ independent batches. The jobs in the last batch are significantly larger than all the previous ones.
We construct the jobs in all first $(\|\mathcal{B}\|^2 + 2)$ batches so that the order of execution of the jobs in the last batch depends on whether there exists a subset $S' \subset S$ such that $\sum_{x_i \in S'}x_i = x$. If the subset does not exist the organizations are prioritized in some predefined order $\sigma_{ord}$; otherwise, the order is reversed $\sigma_{rev}$. The sizes of the jobs in the last batch are so large that they dominate the values of the utilities of the organizations. The relative distance between the utilities in $\sigma_{ord}$ and in $\sigma_{rev}$ is $(1-\epsilon)$ so any $(\frac{1}{2} - \epsilon)$-approximation algorithm $\mathcal{A}$ would allow to infer the true fair schedule for such constructed instance, and so the answer to the initial $\textsc{SubsetSum}$ problem. The precise construction is described below.

We show that if there is an $(\frac{1}{2} - \epsilon)$-approximation algorithm $\mathcal{A}$ for calculating the vector of the contributions, then we would be able to use $\mathcal{A}$ for solving the $\textsc{SubsetSum}$ problem (which is NP-hard). This proof is similar in a spirit to the proof of Theorem~\ref{thm::contributionsNp}. Let $I$ be an instance of the $\textsc{SubsetSum}$ problem, in which we are given a set $S = \{x_1, x_2, \dots, x_k\}$ of $k$ integers and a value $x$. In the $\textsc{SubsetSum}$ problem we ask for the existence of a subset $S' \subset S$ such that $\sum_{x_i \in S'}x_i = x$; we will call the subsets $S' \subset S$ such that $\sum_{x_i \in S'}x_i = x$ the $x$-sum subsets.

From $I$ we construct the instance of the problem of calculating the vector of contributions in the following way. We set $\orgs = \orgs_{S} \cup \{a\} \cup \mathcal{B}$ to be the set of all organizations where  $\orgs_{S} = \{O_1, \dots, O_k\}$ ($\|\orgs_{S}\| = k$) is the set of the organizations corresponding to the appropriate elements of $S$ and $\{a\} \cup \mathcal{B}$, where $\mathcal{B} = \{B_1, \dots, B_{\ell}\}$ ($\ell = \|\mathcal{B}\|$ will be defined afterwards; intuitively $\ell \gg k$), is the set of dummy organizations needed for our construction.

We divide the time into $(\|\mathcal{B}\|^2 + 3)$ independent batches. The batches are constructed in such a way that the $(j+1)$-th batch starts after the time in which all the jobs released in $j$-th batch are completed in every coalition (thus, the duration of the batch can be just the maximum release time plus the sum of the processing times of the jobs released in this batch). As the result, the contribution $\phi^{(u)}$ of each organization $O^{(u)}$ is the sum of its contributions in the all $(\|\mathcal{B}\|^2 + 3)$ batches. For the sake of the clarity of the presentation we assume that time moments in each batch are counted from 0.

We start from the following observation: if the sum of the processing times of the jobs in a batch is equal to $p_{sum}$, then the contribution of each organization can be upper bounded by $p_{sum}^2$. This observation follows from the fact that any organization, when joining a coalition, cannot decrease the completion time of any job by more than $p_{sum}$. As the total number of unit-size parts of the jobs is also $p_{sum}$, we infer that the joining organization cannot increase the value of the coalition by more than $p_{sum}^2$. The second observation is the following: if the joining organization causes decrease of the completion time of the task with processing time $p$, then its contribution is at least equal to $\frac{p}{\|\orgs\|!}$ (as it must decrease the start time of the job by at least one time unit in at least one coalition).

Let $x_{tot} = \sum_{j = 1}^{k}x_{j}$. In our construction we use 4 large numbers $L, XL, H$ and $XH$, where $L = (\|\orgs\| + 1 + 4\|\mathcal{B}\|^2x_{tot}^{2})\cdot \|\orgs\|!$; $XL = (\orgs! \cdot L \cdot \|\orgs\|(\|\orgs\| + 1))^2 + \orgs!4\|\mathcal{B}\|^2x_{tot}^{2} + 1$, $H = \|\mathcal{B}\|^2(2\|\orgs\|(1 + x_{tot}) + 2x + XL)^2 + 1$ and $XH$ is a very large number that will be defined afterwards. Intuitively: $XH \gg H \gg XL \gg L \gg x_{tot}$.

%We start from describing the jobs in the last batch. In $(\|\mathcal{B}\|^2 + 3)$-th batch only the organizations from $\mathcal{B}$ release their jobs. Each such organization releases $\|\mathcal{O}\|$ jobs in time $0$, each of size $XH$.

In the first batch only the organizations from $\mathcal{B}$ release their jobs. The $i$-th organization from $\mathcal{B}$ releases $2i$ jobs in time $0$, each of size $L$. This construction is used to ensure that after the first batch the $i$-th organization from $\mathcal{B}$ has the difference $(\phi^{(i)} - \psi^{(i)})$ greater than the difference $(\phi^{(i+1)} - \psi^{(i+1)})$ of the $(i+1)$-th organization from $\mathcal{B}$ of at least $\frac{L}{\|\orgs\|!} = (\|\mathcal{O}\| + 1 + 4\|\mathcal{B}\|^2x_{tot}^{2})$ and of at most $p_{sum}^2 = (L \cdot \frac{\|\mathcal{B}\|(\|\mathcal{B}\| + 1)}{2})^2 < \frac{XL}{\orgs!} - 4\|\mathcal{B}\|^2x_{tot}^{2}$.

In the second batch, at time $0$, all the organizations except for $a$ release 2 jobs, each of size $H$. This construction is used to ensure that after the second batch the contribution (and so the the difference ($\phi - \psi$)) of the organization $a$ is large (at least equal to $H$, as $a$ joining any coalition causes the job of size $H$ to be scheduled at least one time unit earlier). Since in each of the next $\|\mathcal{B}\|^2$ batches the total size of the released jobs will be lower than $(2\|\orgs\|(1 + x_{tot}) + 2x + XL)$, we know that in each of the next $\|\mathcal{B}\|^2$ batches the jobs of $a$ will be prioritized over the jobs of the other organizations.

Each of the next $\|\mathcal{B}\|^2$ batches is one of the $2\|\mathcal{B}\|$ different types. For the organization $B_i$ ($1 \leq i \leq \|\mathcal{B}\|$) there is exactly $i$ batches of type \texttt{Bch}($B_i$, $2x+1$) and $(\|\mathcal{B}\| - i)$ batches of type \texttt{Bch}($B_i$, $2x$). The order of these $\|\mathcal{B}\|^2$ batches can be arbitrary.

The batches \texttt{Bch}($B_i$, $2x$) and \texttt{Bch}($B_i$, $2x+1$) are similar. The only difference is in the jobs of the organization $a$. In the batch \texttt{Bch}($B_i$, $2x$) the organization $a$ has two jobs $J^{(a)}_1$ and $J^{(a)}_2$, with release times $r^{(a)}_1 = 0$ and $r^{(a)}_2 = 2x$ and processing times $p^{(a)}_1 = 2x+1$ and $p^{(a)}_2 = XL$. 
In the batch \texttt{Bch}($B_i$, $2x+1$) the organization $a$ has two jobs $J^{(a)}_1$ and $J^{(a)}_2$, with release times $r^{(a)}_1 = 0$ and $r^{(a)}_2 = 2x+1$ and processing times $p^{(a)}_1 = 2x+2$ and $p^{(a)}_2 = XL$. All other organizations have the same jobs in batches \texttt{Bch}($B_i$, $2x$) and \texttt{Bch}($B_i$, $2x+1$). The organization $B_i$ has no jobs and all the other organizations from $\mathcal{B}$ release a single job of size $(2x_{tot} + 2)$ in time $0$. The $j$-th organization from $\orgs_{S}$ has two jobs $J^{(j)}_1$ and $J^{(j)}_2$, with release times $r^{(j)}_1 = 0$ and $r^{(j)}_2 = 1$ and processing times $p^{(j)}_1 = 2x_{tot} + 1$ and $p^{(j)}_2 = 2x_j$.

Finally, in the last $(\|\mathcal{B}\|^2 + 3)$-th batch only the organizations from $\mathcal{B}$ release their jobs. Each such organization releases $\|\mathcal{O}\|$ jobs in time $0$, each of size $XH$.

\begin{figure}[h!p]
  \begin{center}
	\hspace{-7cm}
    \includegraphics[trim=5cm 14cm 0cm 5cm, scale=0.65]{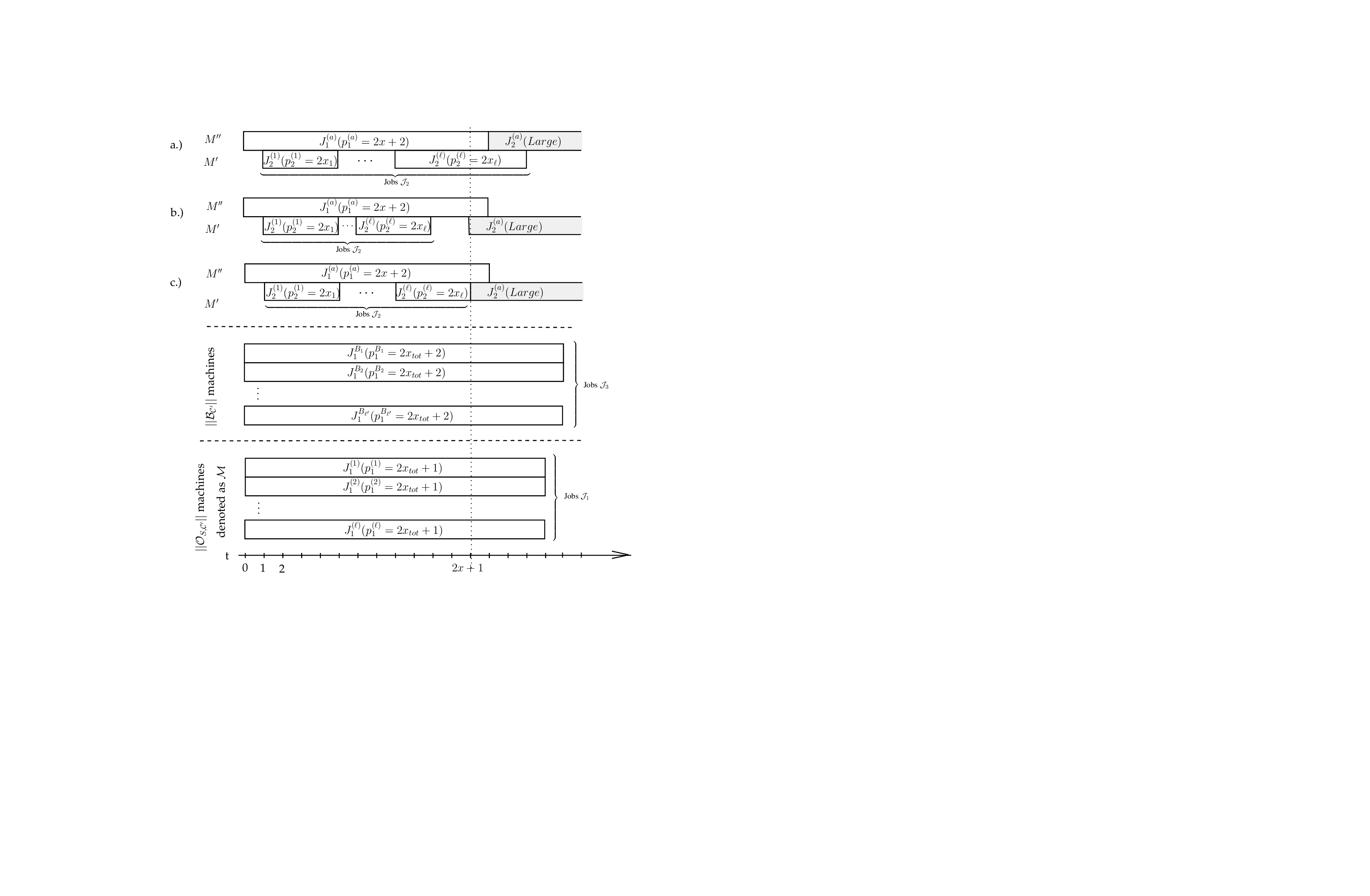}
  \end{center}
  %\vspace{-9.0cm}
  \caption[Figure for inapproximability proof]{The schedule for the coalition $\mathcal{C}'$ such that $B_i \in \mathcal{C}'$ and $a \in \mathcal{C}'$ in batch \texttt{Bch}($B_i$, $2x+1$), for 3 cases: a)~$\sum_{x_i: O_i \in \orgs_{S, C'}}x_i > x$, b)~$\sum_{x_i: O_i \in \orgs_{S, C'}}x_i < x$, c)~$\sum_{x_i: O_i \in \orgs_{S, C'}}x_i = x$. 
We compare $B_i$'s contribution $\phi$ on this schedule to schedule \texttt{Bch}($B_i$, $2x$) 
(not shown; the only differences are that $p_1^{(a)}=2x+1$ and $r_2^{(a)}=2x$).
Other organizations $B_j \neq B_i$ have utility equal to contribution in all cases considered here.
As $B_i$ has no jobs, it contributes only a single machine (corresponding to $M'$). 
Thanks to $M'$, the small jobs $J_2^{(i)}$ execute at most $2x_{tot}$ earlier (if there is no machine $M'$, these jobs are executed at $\mathcal{M}$).
The total size of these small jobs is $2x_{tot}$. Regarding small jobs, the resulting contribution of $B_i$ to $\mathcal{C}'$ is bounded by $4x_{tot}^{2}$.

\hspace{2em}
In case a) $M'$ does not decrease the start time of the large job $J_2^{(a)}$;
the same happens in batch \texttt{Bch}($B_i$, $2x$). 
In case b) $M'$ speeds up $J_2^{(a)}$ by 1; the same happens in batch \texttt{Bch}($B_i$, $2x$). 
In case c) $M'$ also speeds up $J_2^{(a)}$ by 1; however, in batch \texttt{Bch}($B_i$, $2x$) $M'$ does not decrease $J_2^{(a)}$'s start time ($J_2^{(a)}$ is always started at $(2x+1)$).
To summarize, $B_i$ contribution to $\mathcal{C}'$
in both a) and b)
differs by at most $4x_{tot}^{2}$ between \texttt{Bch}($B_i$, $2x$) and \texttt{Bch}($B_i$, $2x+1$).
In contrast, in c)
the contribution in \texttt{Bch}($B_i$, $2x+1$) is greater by at least $XL - 4x_{tot}^{2}$ compared to the contribution in \texttt{Bch}($B_i$, $2x$). 

\hspace{2em}
As a consequence, considering $B_i$'s contribution to all coalitions,
if there exists an $x$-sum subset $S' \subset S$ (case c), then the contribution of $B_{i}$ in \texttt{Bch}($B_i$, $2x+1$) is by at least $\frac{XL}{\|\orgs\|!} - 4x_{tot}^{2}$ greater than in \texttt{Bch}($B_i$, $2x$);  if there is no such an $x$-sum subset, then the contribution of $B_{i}$ in \texttt{Bch}($B_i$,~$2x+1$) and in \texttt{Bch}($B_i$,~$2x$) differ by no more than $4x_{tot}^{2}$.}
  \label{fig:nonApprox}
\end{figure}

Now let us compare the schedules for the batches \texttt{Bch}($B_i$, $2x$) and \texttt{Bch}($B_i$, $2x+1$) (see Figure~\ref{fig:nonApprox}).
Let us consider a schedule for a coalition $\mathcal{C}'$. Let  $\orgs_{S, \mathcal{C}'} = \orgs_{S} \cap \mathcal{C}'$; let $\mathcal{B}_{\mathcal{C}'} = \mathcal{B} \cap \mathcal{C}' \setminus \{B_i\}$.
Let $\mathcal{J}_1$ denote the set of $\|\orgs_{S, \mathcal{C}'}\|$ jobs of sizes $2x_{tot} + 1$ (these are the first jobs of the organizations from $\orgs_{S, \mathcal{C}'}$). Let $\mathcal{J}_2$ denote the set of $\|\orgs_{S, \mathcal{C}'}\|$ jobs of sizes from $S$ (the second jobs of the organizations from $\orgs_{S, \mathcal{C}'}$). Let $\mathcal{J}_3$ denote the $\|\mathcal{B}_{\mathcal{C}'}\|$ jobs of the organizations from $\mathcal{B}_{\mathcal{C}'}$ of sizes $2x_{tot} + 2$ (the single jobs of these organizations).

If $\sum_{x_i: O_i \in \orgs_{S, C'}}x_i > x$ or $\sum_{x_i: O_i \in \orgs_{S, C'}}x_i < x$ the schedules for any $\mathcal{C}'$ in batches \texttt{Bch}($B_i$, $2x$) and \texttt{Bch}($B_i$, $2x+1$) looks similarly.
In time $0$, $\|\orgs_{S, \mathcal{C}'}\|$ machines will schedule the $\|\orgs_{S, \mathcal{C}'}\|$ jobs from $\mathcal{J}_1$ (let us denote these machines as $\mathcal{M}$) and $\|\mathcal{B}_{\mathcal{C}'}\|$ machines will schedule the $\|\mathcal{B}_{\mathcal{C}'}\|$ jobs from $\mathcal{J}_3$. If $B_i \notin \mathcal{C}'$ then the jobs from $\mathcal{J}_2$ will be scheduled on the machines from $\mathcal{M}$ just after the jobs from $\mathcal{J}_1$. If $B_i \in \mathcal{C}'$ and $a \notin \mathcal{C}'$, then the coalition $\mathcal{C}'$ has $(\|\orgs_{S, \mathcal{C}'}\| + \|\mathcal{B}_{\mathcal{C}'}\| + 1)$ machines; one machine will execute the jobs from $\mathcal{J}_2$. If $B_i \in \mathcal{C}'$ and $a \in \mathcal{C}'$ then the coalition $\mathcal{C}'$ has $(\|\orgs_{S, \mathcal{C}'}\| + \|\mathcal{B}_{\mathcal{C}'}\| + 2)$ machines. One machine (denoted as $M'$) will execute the jobs from $\mathcal{J}_2$ and one other machine (denoted as $M''$) will execute the job $J^{(a)}_1$. Now, if $\sum_{x_i: O_i \in \orgs_{S, C'}}x_i < x$ then $J^{(a)}_2$ will be scheduled on $M'$; otherwise on $M''$ (this follows from the construction in the second batch -- we recall that the jobs of $a$ should be prioritized). Thus, as explained in Figure~\ref{fig:nonApprox}, if $\sum_{x_i: O_i \in \orgs_{S, C'}}x_i > x$ or $\sum_{x_i: O_i \in \orgs_{S, C'}}x_i < x$ the contribution and the utility of each organization from $\mathcal{B}$ in two batches \texttt{Bch}($B_i$, $2x$) and \texttt{Bch}($B_i$, $2x+1$) differ by at most $4x_{tot}^{2}$.

If $\sum_{x_i: O_i \in \orgs_{S, C'}}x_i = x$, then the schedules for the cases: (i) $B_i \notin \mathcal{C}'$ (ii) ($B_i \in \mathcal{C}'$ and $a \notin \mathcal{C}'$) remain the same as in case $\sum_{i: O_i \in \orgs_{S, C'}}x_i \neq x$. For the last case ($B_i \in \mathcal{C}'$ and $a \in \mathcal{C}'$) the jobs from $\mathcal{J}_1$, from $\mathcal{J}_2$ and $J^{(a)}_1$ are scheduled in the same way as previously. However, the job $J^{(a)}_2$ will be scheduled in \texttt{Bch}($B_i$, $2x+1$) on machine $M'$ (in the moment it is released) and in \texttt{Bch}($B_i$, $2x$) on machine $M'$ or $M''$ (one time unit later than it was released). As explained in Figure~\ref{fig:nonApprox}, if there exists an $x$-sum subset $S' \subset S$, then the contribution of $B_i$ in \texttt{Bch}($B_i$, $2x+1$) will be greater by at least of $\frac{XL}{\orgs!} - 4x_{tot}^{2}$ than in \texttt{Bch}($B_i$, $2x$).

As the result, if there does not exist an $x$-sum subset $S' \subset S$, then the difference ($\phi^{(i)} - \psi^{(i)}$) for the $i$-th organization from $\mathcal{B}$ will be greater than the difference ($\phi^{(i+1)} - \psi^{(i+1)}$) for the $(i+1)$-th organization from $\mathcal{B}$ by at least $(\|\mathcal{O}\| + 1)$ (from the construction in the first batch the difference $(\phi^{(i)} - \psi^{(i)})$ was greater than ($\phi^{(i+1)} - \psi^{(i+1)}$) by at least $(\|\mathcal{O}\| + 1 + 4\|\mathcal{B}\|^2x_{tot}^{2})$, and as explained in Figure~\ref{fig:nonApprox} the difference $(\phi^{(i+1)} - \psi^{(i+1)} - \phi^{(i)} + \psi^{(i)})$ could change by at most $4\|\mathcal{B}\|^2x_{tot}^{2}$). Otherwise, the difference ($\phi^{(i)} - \psi^{(i)}$) for the $i$-th organization will be lower than for the $(i+1)$-th organization (as there are more batches of type \texttt{Bch}($B_{i+1}$, $2x+1$) than of type \texttt{Bch}($B_i$, $2x+1$)).

Thus, if there does not exist an $x$-sum subset $S' \subset S$, then in the last batch the jobs of $B_1$ will be scheduled first, than the jobs of $B_2$, and so on -- let us denote such schedule as $\sigma_{ord}$. On the other hand, if there exists an $x$-sum subset $S' \subset S$, the jobs in the last batch will be scheduled in the exactly reverse order -- such schedule will be denoted as $\sigma_{rev}$.

Now, let us assess the distance between the vector of utilities in case of two schedules $\sigma_{ord}$ and $\sigma_{rev}$. Let us assume that $\|\mathcal{B}\|$ is even. Every job of the organization $B_{i}$ ($1 \leq i \leq \frac{\|\mathcal{B}\|}{2}$) in the last batch is started $XH(\|\mathcal{B}\|- 2i + 1)$ time units earlier in $\sigma_{ord}$ than in $\sigma_{rev}$. The jobs of the organization $B_{(\|\mathcal{B}\| + 1 - i)}$ ($1 \leq i \leq \frac{\|\mathcal{B}\|}{2}$) are scheduled $XH(\|\mathcal{B}\|- 2i + 1)$ time units later in $\sigma_{ord}$ than in $\sigma_{rev}$. Since each such job consists of $XH$ unit-size elements, the distance between the vector of utilities for $\sigma_{ord}$ and $\sigma_{rev}$, denoted as $\Delta\psi$, can be lower bounded by:
\begin{align*}
\Delta\psi \geq 2\|\orgs\|\sum_{i=1}^{\|\mathcal{B}\|/2} (2i - 1)XH^2
=\|\orgs\|\|\mathcal{B}\|\frac{(1 + 1 + 2\frac{\|\mathcal{B}\|}{2} -2)}{2}XH^2 = \frac{1}{2}\|\orgs\| \|\mathcal{B}\|^2XH^2\textrm{.}
\end{align*}
\noindent Now, we can define $XH$ to be the total size of the all except the last batch times $\frac{4}{\epsilon}$.
Below we show how to bound the total utility $\psi_{tot}$ of the true fair schedule ($\sigma_{ord}$ or $\sigma_{rev}$) in the time $t$ when all the jobs are completed. Each unit size part of the job completed in time $t$ contributes to the utility the value 1.  Each unit size part of the job executed in time $t-1$ is worth 2, and so on. Since the jobs in the last batch are executed on $\|\orgs\|$ machines and the duration of the batch is equal to $\|\mathcal{B}\|XH$, the utility of the jobs from the last batch is equal to $\sum_{i=1}^{\|\mathcal{B}\|XH} i$. The jobs in all previous batches are started no earlier than in $t - \|\mathcal{B}\|XH - \frac{\epsilon}{4}XH$. The duration of the all but the last batch can be upper bounded by $\frac{\epsilon}{4}XH$. There are $\|\orgs\|$ machines, so the utility of the jobs from the all but the last batch can be upper bounded by $(\|\mathcal{B}\|XH + \frac{\epsilon}{4}XH)\frac{\epsilon}{4}XH$. Thus we get the following bound on $\psi_{tot}$:
\begin{align*}
\psi_{tot} &< \|\orgs\|\left(\sum_{i=1}^{\|\mathcal{B}\|XH} i + \left(\|\mathcal{B}\|XH + \frac{\epsilon}{4}XH\left)\right(\frac{\epsilon}{4}XH\right)\right) \\
&\leq \|\orgs\|\left(\frac{1 + \|\mathcal{B}\|XH}{2}\|\mathcal{B}\|XH + \frac{\epsilon}{4}\|\mathcal{B}\|XH^2 + \frac{\epsilon}{16}^2XH^2 \right) \\ 
&\leq \|\orgs\|\left(\frac{1}{2}\left(1 + \|\mathcal{B}\|XH \right)^2 + \frac{\epsilon}{4}\|\mathcal{B}\|^2XH^2 \right) \\
&\leq \|\orgs\| \|\mathcal{B}\|^2XH^2 \left(\frac{1}{2} \cdot \left(\frac{1 + \|\mathcal{B}\|}{\|\mathcal{B}\|}\right)^2 + \frac{\epsilon}{4}\right) \textrm{.}
\end{align*}

We can chose the size $\|\mathcal{B}\|$ so that $\left(\frac{1+\|\mathcal{B}\|}{\|\mathcal{B}\|}\right)^2 < 1 + \frac{\epsilon}{2}$. As the result we have:
\begin{align*}
\Delta\psi / \psi_{tot} >  \frac{1}{2} / \frac{1}{2}\left(\left(\frac{1 + \|\mathcal{B}\|}{\|\mathcal{B}\|}\right)^2 + \frac{\epsilon}{2}\right) > \frac{1}{1+\epsilon} > 1 - \epsilon
\end{align*}

Finally let us assume that there exists  $(\frac{1}{2} - \epsilon)$-approximation algorithm $\mathcal{A}$ that returns the schedule $\sigma$ for our instance. Now, if $\sigma$ is closer to $\sigma_{ord}$ than 
to $\sigma_{rev}$, we can infer that $\sigma_{ord}$ is a true fair solution to our instance
(and so the answer to the initial \textsc{SubsetSum} question is ``yes''). Otherwise, $\sigma_{rev}$ is a true solution (and the answer to the \textsc{SubsetSum} problem is ``no''). This completes the proof.

\end{proof}

\subsection{Special case: unit-size jobs}\label{sec::unitSize}

\begin{figure}[th!]
\begin{algorithm}[H]
   \footnotesize
   \SetKwFunction{Distance}{Distance}
   \SetKwInput{KwNotation}{Notation}
   \SetKwFunction{ReleaseJob}{ReleaseJob}
   \SetKwFunction{Prepare}{Prepare}
   \SetKwFunction{FairAlgorithm}{FairAlgorithm}
   \SetKwFunction{UpdateVals}{UpdateVals}
   \SetKwFunction{PropoagateVals}{PropoagateVals}
   \SetKwFunction{FreeMachine}{FreeMachine}
   \SetKwFunction{SelectAndSchedule}{SelectAndSchedule}
   \SetKwBlock{Block}
   \SetAlCapFnt{\footnotesize}
   \KwNotation{\\
	$\epsilon$, $\lambda$ --- as in Theorem~\ref{thm::approxRandomized}}

	\vspace{0.3cm}\Prepare{$\mathcal{C}$}:
	\Block{
		$N \leftarrow \lceil \frac{\|\mathcal{C}\|^2}{\epsilon^2}\ln\left(\frac{\|\mathcal{C}\|}{1 - \lambda}\right) \rceil$\;
		$\Gamma \leftarrow$ generate $N$ random orderings (permutations) of the set of all organizations (with replacement)\;
		$Subs \leftarrow Subs' \leftarrow \emptyset$ \;
		\ForEach{$\prec \in \Gamma$}
		{
			\For{$u \leftarrow 1$ \KwTo $\|C\|$}
			{
				$\mathcal{C'} \leftarrow \{O^{(i)}: O^{(i)} \prec O^{(u)}\}$ \;
				$Subs \leftarrow Subs \cup \{\mathcal{C'}\}$; $Subs' \leftarrow Subs' \cup \{\mathcal{C'} \cup \{O^{(u)}\}\}$ \;
			}
		}
	}

        \vspace{0.3cm}\ReleaseJob{$O^{(u)}$, $J$}:
	\Block{
		\For{$\mathcal{C}' \in Subs \cup Subs': O^{(u)} \in \mathcal{C}'$}
		{
			$\mathrm{jobs}[\mathcal{C}'][O^{(u)}].\mathrm{push}(J)$
		}
	}

	\vspace{0.3cm}\SelectAndSchedule{$\mathcal{C}$, $t$}:
	\Block{
		$u \leftarrow \mathrm{argmin}_{O^{(u)}}(\psi[\mathcal{C}][O^{(u)}] - \phi[\mathcal{C}][O^{(u)}])$ \;
		$\sigma[\mathcal{C}] \leftarrow \sigma[\mathcal{C}] \cup \{(\mathrm{jobs}[\mathcal{C}][u].\mathrm{first}, t)\}$\;
		$\mathrm{finPerOrg}[O^{(u)}] \leftarrow \mathrm{finPerOrg}[O^{(u)}] + 1$\;
		$\phi[O^{(u)}] \leftarrow \phi[O^{(u)}] + 1$\;
	}

	\vspace{0.3cm}\FairAlgorithm{$\mathcal{C}$}:
	\Block{
		\Prepare{$\mathcal{C}$} \;
		\ForEach{time moment $t$}
		{
			\ForEach{job $J_i^{(u)}$: $r_i^{(u)} = t$}
			{
				\ReleaseJob{$O_i^{(u)}, J_i^{(u)}$}\;
			}
			\ForEach{$\mathcal{C'} \subset Subs \cup Subs'$}
			{\nllabel{algline::line3}
				$\mathrm{v}[\mathcal{C'}] \leftarrow \mathrm{v}[\mathcal{C'}] + \mathrm{finPerCoal}[\mathcal{C'}]$ \;
				$n \leftarrow \min(\sum_{O^{(u)} \in \mathcal{C'}} m^{(u)}, \|\mathrm{jobs}[\mathcal{C}][O^{(u)}]\|)$ \;
				remove first $n$ jobs from $\mathrm{jobs}[\mathcal{C}][O^{(u)}]$ \;
				$\mathrm{finPerCoal}[\mathcal{C'}] \leftarrow \mathrm{finPerCoal}[\mathcal{C'}] + n$ \;
				$\mathrm{v}[\mathcal{C'}] \leftarrow \mathrm{v}[\mathcal{C'}] + n$ \;
			}

			\ForEach{$O^{(u)} \in \mathcal{C}$}
			{\nllabel{algline::line4}
				$\psi[O^{(u)}] \leftarrow \psi[O^{(u)}] + \mathrm{finPerOrg}[O^{(u)}]$\;
				$\phi[O^{(u)}] \leftarrow 0$\;
				\ForEach{$\mathcal{C}' \in Subs: O^{(u)} \notin \mathcal{C}'$}
				{\nllabel{algline::line5}
					$\mathrm{marg}\_\phi \leftarrow \mathrm{v}[\mathcal{C'} \cup \{O^{(u)}\}] - \mathrm{v}[\mathcal{C'}]$ \;
					$\phi[O^{(u)}] \leftarrow \phi[O^{(u)}] + \mathrm{marg}\_\phi \cdot \frac{1}{N}$\;
				}
			}

			\While{\FreeMachine{$\sigma[\mathcal{C}]$, $t$}}
			{
				\SelectAndSchedule{$\mathcal{C}, t$}\;
			}
		}
	}

\end{algorithm}
\caption{Algorithm~\textsc{Rand}: a fair algorithm for the specific utility function $\psi_{sp}$ and for unit-size jobs.}\label{alg:fairAlgUnitSize}
\end{figure}

In case when the jobs are unit-size the problem has additional properties that allow us to construct an efficient approximation (however, the complexity of this special case is open). However, the results in this section do not generalize to related or unrelated processors. For unit-size jobs, the value of each coalition $v(\mathcal{C})$ does not depend on the schedule:

\begin{proposition}\label{prop::independentOnAlg}
For any two greedy algorithms $\mathcal{A}_1$ and $\mathcal{A}_2$, for each coalition $\mathcal{C}$ and each time moment $t$, the values of the coalitions $v(\mathcal{A}_1, \mathcal{C}, t)$ and $v(\mathcal{A}_2, \mathcal{C}, t)$ are equal, provided all jobs are unit-size.
\end{proposition}
\begin{proof}
We prove the following stronger thesis: for every time moment $t$ any two greedy algorithms $\mathcal{A}_1$ and $\mathcal{A}_2$ schedule the same number of the jobs till $t$. We prove this thesis by induction. The base step for $t=0$ is trivial. Having the thesis proven for $(t-1)$ and, thus knowing that in $t$ in both schedules there is the same number of the jobs waiting for execution (here we use the fact that the jobs are unit-size), we infer that in $t$ the two algorithms schedule the same number of the jobs. Since the value of the coalition does not take into account the owner of the job,  we get the thesis for $t$. This completes the proof.
\end{proof}

As the result, we can use the randomized approximation algorithm for the scheduling problem restricted to unit-size jobs (Algorithm~\textsc{Rand} from Figure~\ref{alg:fairAlgUnitSize}). The algorithm is inspired by the randomized approximation algorithm for computing the Shapley value presented by Liben-Nowell et al~\cite{conf/cocoon/Liben-NowellSWW12}. However, in our case, the game is not supermodular (which is shown in Proposition~\ref{prop:notSupermodular} below), and so we have to adapt the algorithm and thus obtain different approximation bounds.

\begin{proposition}\label{prop:notSupermodular}
In case of unit-size jobs the cooperation game in which the value of the coalition $\mathcal{C}$ is defined by  $v(\mathcal{C}) = \sum_{O^{(u)} \in \mathcal{C}} \psi(O^{(u)})$ is not supermodular.
\end{proposition}
\begin{proof}
Consider a following instance with 3 organizations: $a$, $b$ and $c$ each owning a single machine. Organizations $a$ and $b$ in time $t = 0$ release two unit size jobs each; the organization $c$ has no jobs. We are considering the values of the coalitions in time $t = 2$; $v(\{a, c\}) = 4$ (the two jobs are scheduled in time 0), $v(\{b, c\}) = 4$, $v(\{a, b, c\}) = 7$ (three jobs are scheduled in time 0 and one in time 1) and $v(\{c\})$ = 0 (there is no job to be scheduled). We see that  $v(\{a, b, c\})$ + $v(\{c\})$ $<$ $v(\{a, c\})$ + $v(\{b, c\})$, which can be written as:
\begin{align*}
v(\{a, c\} \cup \{b, c\}) + v(\{a, c\} \cap \{b, c\}) < v(\{a, c\}) + v(\{b, c\}) \textrm{.}
\end{align*}
This shows that the game is not supermodular.
\end{proof}

In this algorithm we keep simplified schedules for a random subset of all possible coalitions. 
For each organization $O^{(u)}$ the set $Subs[O^{(u)}]$ keeps $N = \frac{\|\mathcal{C}\|^2}{\epsilon^2}\ln\left(\frac{\|\mathcal{C}\|}{1 - \lambda}\right)$ random coalitions not containing $O^{(u)}$; for each such random coalition $\mathcal{C}'$ which is kept in $Subs[O^{(u)}]$, $Subs'[O^{(u)}]$ contains the coalition $\mathcal{C}' \cup \{O^{(u)}\}$. For the coalitions kept in $Subs[O^{(u)}]$ we store a simplified schedule (the schedule that is determined by an arbitrary greedy algorithm). The simplified schedule allows us to find the value $v(\mathcal{C}')$ of the coalition $\mathcal{C}'$. Maintaining the whole schedule would require the recursive information about the schedules in the subcoalitions of $\mathcal{C}'$. However, as the consequence of Proposition~\ref{prop::independentOnAlg} we know that the value of the coalition $v(\mathcal{C}')$ can be determined by an arbitrary greedy algorithm\footnote{In this point we use the assumption about the unit size of the jobs. The algorithm cannot be extended to the general case. In a general case, for calculating the value for each subcoalition we would require the exact schedule which cannot be determined polynomially (Theorem~\ref{thm::contributionsNp}).}.

The third $foreach$ loop in procedure \texttt{FairAlgorithm} (line~\ref{algline::line3} in Figure~\ref{alg:fairAlgUnitSize}) updates the values of all coalitions kept in $Subs$ and $Subs$'. From Equation~\ref{def::fairMetric} it follows that after one time unit if no additional job is scheduled, the value of the coalition increases by the number of completed unit-size parts of the jobs (here, as the jobs are unit size, the number of the completed jobs is $\mathrm{finPerCoal}[\mathcal{C'}]$). In time moment $t$, all waiting jobs 
(the number of such jobs is $\|\mathrm{jobs}[\mathcal{C}][O^{(u)}]\|$) are scheduled provided there are enough processors (the number of the processors is $\sum_{O^{(u)} \in \mathcal{C'}} m^{(u)}$). If $n$ additional jobs are scheduled in time $t$ then the value of the coalition in time $t$ increases by $n$. 

In the fourth $foreach$ loop (line~\ref{algline::line4} in Figure~\ref{alg:fairAlgUnitSize}), once again we use the fact that the utility of the organization after one time unit increases by the number of finished jobs ($\mathrm{finPerOrg}[O^{(u)}]$). In the last $foreach$ loop (line~\ref{algline::line5}) the contribution of the organization is approximated by summing the marginal contributions $\mathrm{marg}\_\phi$ only for the kept coalitions. Theorem~\ref{thm::approxRandomized} below gives the bounds for the quality of approximation.

\begin{theorem}\label{thm::approxRandomized}
Let $\vec{\psi}$ denote the vector of utilities in the schedule determined by Algorithm~\textsc{Rand} from Figure~\ref{alg:fairAlgUnitSize}. If the jobs are unit-size, then $\mathcal{A}$ with the probability $\lambda$ determines the $\epsilon$-approximation schedule, i.e. gives guarantees for the bound on the distance to the truly fair solution:
\begin{align*}
\|\vec{\psi} - \vec{\psi}^{*}\|_{M} \leq \epsilon|\vec{\psi}^{*}| \textrm{.}
\end{align*}
\end{theorem}
\begin{proof}
Let us consider an organization $O^{(u)}$ participating in a coalition $\mathcal{C}$ and a time moment $t$. Let $\phi^{(u), *}$ and $\psi^{(u),*}$ denote the contribution and the utility of the organization $O^{(u)}$ in a coalition $\mathcal{C}$ in time moment $t$ in a truly fair schedule. Let $v^{*}(\mathcal{C})$ denote the value of the coalition $\mathcal{C}$ in a truly fair schedule. According to notation in Figure~\ref{alg:fairAlgUnitSize}, let $\phi[O^{(u)}]$ and $\psi[O^{(u)}]$ denote the contribution and the utility of the organization $O^{(u)}$ in a coalition $\mathcal{C}$ in time $t$ in a schedule determined by Algorithm~\textsc{Rand}; Let $N = \frac{\|\mathcal{C}\|^2}{\epsilon^2}\ln\left(\frac{\|\mathcal{C}\|}{1 - \lambda}\right)$. First, note that $|\psi^{(u),*} - \psi[O^{(u)}]| \leq |\phi^{(u),*} - \phi[O^{(u)}]|$. Indeed, if the contribution of the organization $O^{(u)}$ increases by a given value $\Delta \phi$ then Algorithm~\textsc{Rand} will schedule $\Delta \phi$ more unit-size jobs of the organization $O^{(u)}$ provided there is enough such jobs waiting for execution.

Let $X$ denote the random variable that with the probability $\frac{1}{\|\mathcal{C}\|!}$ returns the marginal contribution of the organization $O^{(u)}$ to the coalition composed of the organizations preceding $O^{(u)}$ in the random order (of course, there is $\|\mathcal{C}\|!$ such random orderings). We know that $X \in [0, v^{*}(\mathcal{C})]$ and that $\expected(X) = \phi^{(u), *}$. Algorithm~\textsc{Rand} is constructed in such a way that $\phi[O^{(u)}] = \sum_{i=0}^{N} \frac{1}{N}X_i$, where $X_i$ are independent copies of $X$. Thus, $\expected (\phi[O^{(u)}]) = \phi^{(u), *}$. From Hoeffding's inequality we get the bound on the probability $p_{\epsilon}$ that $\phi[O^{(u)}] - \phi^{(u), *} > \frac{\epsilon}{\|\mathcal{C}\|} v^{*}(\mathcal{C})$:

\begin{align*}
p_{\epsilon} &= \probability\left(\sum_{i=0}^{N} \frac{1}{N}X_i - \phi^{(u), *} > \frac{\epsilon}{\|\mathcal{C}\|} v^{*}(\mathcal{C})\right) \\
             &< \exp\left(-\frac{\epsilon^2v^{*}(\mathcal{C})^2N^2}{v^{*}(\mathcal{C})^2N\|\mathcal{C}\|^2} \right) \\
             &= \exp\left(-\frac{\epsilon^2N}{\|\mathcal{C}\|^2}\right) = \frac{1 - \lambda}{\|\mathcal{C}\|}\textrm{.}
\end{align*}

The probability that $\vec{\phi} - \vec{\phi^{*}} > \epsilon  v^{*}(\mathcal{C})$ can be bounded by $p_{\epsilon}\|\mathcal{C}\| = 1 - \lambda$. As the result, also the probability that $\vec{\psi} - \vec{\psi^{*}} > \epsilon  v^{*}(\mathcal{C})$ can be bounded by $1 - \lambda$, which completes the proof.
\end{proof}

The complexity of Algorithm~\textsc{Rand} is $\|\orgs\|\cdot N = \|\orgs\| \frac{\|\mathcal{C}\|^2}{\epsilon^2}\ln\left(\frac{\|\mathcal{C}\|}{1 - \lambda}\right)$ times the complexity of the single-organization scheduling algorithm. As a consequence, we get the following result:

\begin{theorem}\label{thm:fpras}
There exists an FPRAS for the problem of finding the fair schedule for the case when the jobs are unit size.
\end{theorem}

\section{Resource utilization of greedy algorithms}

It might appear that in order to ensure the fairness of the algorithm we might be forced to use globally inefficient algorithms. Such algorithms might, for instance, waste resources.
We define the resource utilization as the percentage of the time in which, on average, every processor is busy. The resource utilization is an established metric indicating the global efficiency of resource usage. Indeed, even though we use greedy algorithms, some of them might result in suboptimal resource utilization. This problem is shown in Figure~\ref{fig:resourceWaste}.

\begin{figure}[tb]
  \begin{center}
    \includegraphics[scale=1.0]{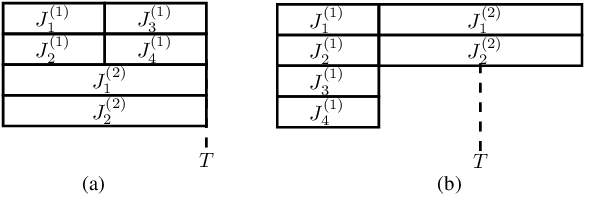}
  \end{center}
  \vspace{-0.15cm}
  \caption{The example showing that greedy algorithms might induce suboptimal resource utilization. In the example we have 4 jobs of the organization $O^{(1)}$, each of size 3, and 2 jobs of the organization  $O^{(2)}$, each of size 6. All the jobs are released in time 0. Let us consider time moment $T = 6$. In Figure (a) the jobs of  $O^{(2)}$ are started first, which results in 100\% resource utilization. In Figure (b)  the jobs of $O^{(1)}$ are started first, which in time $T$ gives 75\% of resource utilization.}\label{fig:resourceWaste}
\end{figure}

Thus, there is a natural question, which in additional to the context of fair scheduling, is interesting on its own. How bad can we be when using a greedy algorithm (with any underlying scheduling policy)? In the next theorem we show that the example from Figure~\ref{fig:resourceWaste} is, essentially, the worst possible scenario.  

\begin{definition}
An algorithm $\mathcal{A}$ is an $\alpha$-competitive online algorithm for resource utilization if and only if in each time moment $T$ the ratio of the resource utilization between the schedule derived by $\mathcal{A}$ and the schedule obtained by any other algorithm is greater or equal to $\alpha$.
\end{definition}

\begin{figure}[tb]
  \hspace{3.0cm} \includegraphics[scale=1.0]{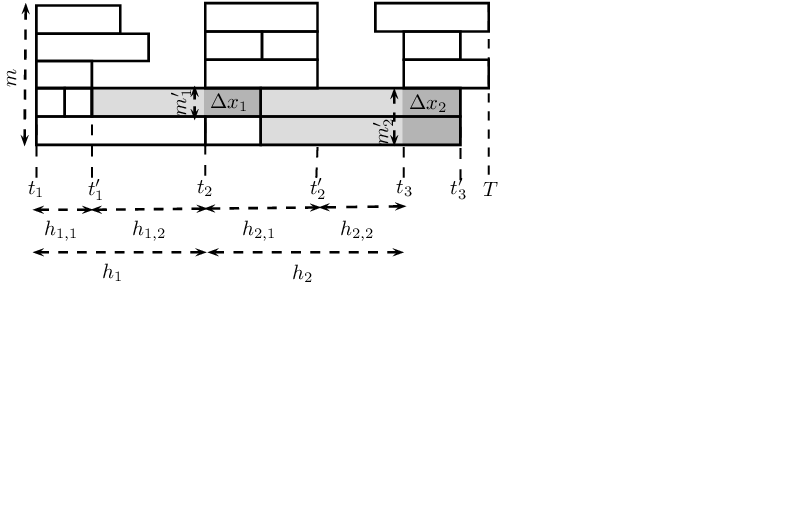}
  \vspace{-3.5cm}
  \caption{The figure illustrating the notation used in the proof of Theorem~\ref{thm:resourceUtil}. There are 3 time blocks in this figure: the first one lasts from $t_1$ till $t_2$; the second from $t_2$ till $t_3$; and the third one from $t_3$ till $T$. The parts of jobs that were delayed outside their time block (in comparison with an optimal schedule) are marked in dark gray. The remaining parts of these jobs are marked in light gray. We see that these jobs which are delayed outside the $i$-th time block are started at or before $t_i'$.}\label{fig:resourceUtil}
\end{figure}

\begin{theorem}\label{thm:resourceUtil}
Every greedy algorithm for scheduling sequential jobs on identical processors is a $\frac{3}{4}$-competitive online algorithm for resource utilization.
\end{theorem}
\begin{proof}
Let $\sigma$ denote the schedule obtained by some greedy algorithm $\mathcal{A}$ until time $T$, and let $\sigma^*$ denote the schedule obtained by the optimal (according to the resource utilization metric) algorithm for the same input. Now, we will divide the time axis into blocks in the following way. The first block starts in time $0$. The $i$-th block ($i > 1$) starts in the earliest possible time moment $t_i$ such that (i) $t_i > t_{i-1}$ (the $i$-th block starts after the $(i-1)$-th one), and (ii) in $t_i$ there are jobs in $\sigma$ running on all the processors and in $(t_{i}-1)$ at least one processor in $\sigma$ is idle. Let $t_{\ell}$ denote the start time of the last block. By convention we take $t_{\ell+1} = T$. The blocks for the example schedule are depicted in Figure~\ref{fig:resourceUtil}. Furthermore, let $t_{i}'$ denote the earliest moment in the $i$-th time block in which some processor is idle. Let $h_i = t_{i+1} - t_i$ denote the duration of the $i$-th time block. Let $h_{i, 1} = t_i' - t_i$ and let $h_{i, 2} = t_{i+1} - t_i'$.

In our proof we will consider the time blocks separately and for each time block we will prove that the total number of the unit-size parts of the jobs completed in this block in schedules $\sigma$ and $\sigma^{*}$ vary by no more than the factor of $\frac{3}{4}$. Throughout this proof we will use the variable $V$ that, intuitively, accumulates the number of unit-size parts of the jobs that in $\sigma$ were completed in the earlier time block than in $\sigma^{*}$. Let $V_{i}$ denote the value of $V$ after we completed an analysis for the $i$-th block, with $V_{0} = 0$.

Let us consider the $i$-th time block. Let $x_i$ and $x_i^{*}$ denote the number of unit-size parts of the jobs completed in the $i$-th time block in schedules $\sigma$ and $\sigma^{*}$, respectively. If $x_i \geq x_i^{*}$ then we increase the variable $V$ by $(x_i - x_i^{*})$. Otherwise, let $\Delta x_i = x_i^{*} - x_i > 0$. We consider the two following cases:
\begin{enumerate}
\item If $V_{i-1} = 0$, then we set $\Delta y_i = \Delta x_i$, and $V_i = 0$.
\item If $V_{i-1} > 0$, then we set $\Delta y_i = \Delta x_i - \min(\Delta x_i, V_{i-1})$, and $V_i = V_{i-1} - \min(\Delta x_i, V_{i-1})$. Intuitively, this means that the unit-size parts of jobs that were computed extra in earlier blocks and accumulated in $V$, pay for some parts of the jobs that were computed in the later block.
\end{enumerate}
Now, if $\Delta y_i = 0$ this means that from $V$ we managed to pay for the parts that, due to inefficiency of the algorithm $\mathcal{A}$, were not computed in the $i$-th time block. Otherwise ($\Delta y_i > 0$), we infer that some $\Delta y_i > 0$ parts of the jobs that were released before $t_{i+1}$ were delayed and in $\sigma$ were not completed in the $i$-th time block (while they were in $\sigma^{*}$). These jobs were released at or after $t_{i}$. Indeed, otherwise the job would be started at time $t_{i}-1$ or earlier (the algorithm $\mathcal{A}$ is greedy, and a processor is idle at $t_{i}-1$), and so, such a job would be processed for the whole duration of the $i$-th time block. Consequently, the unit-size parts of this job would not contribute to $\Delta y_i$ (the number of unit-size parts of this job completed in the $i$-th time block in $\sigma$ would be no greater than in $\sigma^{*}$).

Let us consider the jobs the parts of which contributed to $\Delta y_i$.
Let $m_i'$ denote the number of machines on which these jobs were processed (see Figure~\ref{fig:resourceUtil} for an example). From the pigeonhole principle, at least one from the considered jobs, $J$, was delayed by at least $\frac{\Delta y_i}{m_i'}$. Since the algorithm $\mathcal{A}$ is greedy, $h_{i, 1} \geq \frac{\Delta y_i}{m_i'}$ (there were at least $\frac{\Delta y_i}{m_i'}$ time moments in the $i$-th time block with no idle processors; otherwise $J$ would be started earlier). Also, each from the considered jobs starts in time $t_i'$ at the latest---indeed, if it would be started later, then from the greediness of the algorithm we would infer that the release time of such job is at least $t_i'+1$, and so, such job would not be delayed in $\sigma$. Consequently, through the whole duration of the $i$-th block some $m_i'$ machines are continuously occupied. Thus, the idle surface of the processors in $\sigma$ is at most equal to $ h_{i, 2} (m-m_i')$. Since, $\Delta y_i \leq \Delta x_i$, and $\Delta x_i$ denotes the difference in the number of unit-size parts of the jobs computed in $\sigma^{*}$ and in $\sigma$, we infer that $\Delta y_i \leq h_{i, 2} (m-m_i')$.

Now, let us estimate $o_i$, the number of the occupied slots in the $i$-th block in schedule $\sigma$.
\begin{align*}
o_i \geq h_{i, 1} m + h_{i, 2} m_i' \geq \frac{\Delta y_i}{m_i'}m + \frac{\Delta y_i}{m-m_i'}m_i' \geq \Delta y_i\left(\frac{m}{m_i'} + \frac{m_i'}{m-m_i'}\right) \geq 3\Delta y_i \textrm{.}
\end{align*}
Thus:
$\frac{o_i}{o_i + \Delta y_i} \geq \frac{o_i}{o_i + \frac{1}{3} o_i} = \frac{3}{4}$. 
Since our reasoning can be repeated for every time block, we get the thesis.

\end{proof}

This result shows that even without any information on jobs' release dates, durations even with an arbitrary scheduling policy, we waste no more than $25\%$ of resources. Of course, this loss of efficiency is even smaller when there are many jobs to be computed. For instance, if in any time moment there are jobs waiting for execution, any greedy algorithm achieves 100\% resource utilization.

It is also natural to consider the loss of efficiency according to our strategy-proof metric. We have chosen to consider resource utilization as it has more intuitive meaning in terms of waste of resources. We leave the problem of finding bounds for our strategy-proof metric as a natural follow-up question.

We note that our fair scheduling algorithm is also applicable for parallel jobs (jobs requiring more than one processor). However, for the case of parallel jobs the loss of the global efficiency of an arbitrary greedy algorithm can be higher. 
We leave these extensions, as well as generalization of the processor model to related and unrelated machines, for the future work.

\section{Experimental evaluation of the algorithms}\label{sec::experiments}

In the previous section we showed that the problem of finding a fair schedule is computationally intractable. However, the ideas used in the exponential and the FPRAS algorithms can be used as insights to create reasonable heuristics. In this section we present the experimental evaluation of the fairness of two simple heuristic algorithms, and several algorithms from the scheduling theory.

\subsection{Algorithms}

In this section we describe the algorithms that we evaluate.

\medskip
\noindent
\textbf{\textsc{Ref}}. We used Algorithm~\textsc{Ref} from Figure~\ref{alg:fairAlg} (which is an exponential algorithm) as the referral fair algorithm. 

\medskip
\noindent
\textbf{\textsc{Rand}}. We used Algorithm~\textsc{Rand} from Figure~\ref{alg:fairAlgUnitSize} as a heuristic  for workloads with jobs having different sizes. We verify two versions of the algorithm with $N=15$ and $N=75$ random subcoalitions. \medskip

\begin{figure}[th]
\begin{algorithm}[H]
   \footnotesize
   \SetKwFunction{own}{own}
   \SetKwInput{KwNotation}{Notation}
   \SetKwFunction{startJob}{startJob}
   \SetKwFunction{FairAlgorithm}{FairAlgorithm}
   \SetKwFunction{Schedule}{Schedule}
   \SetKwFunction{Initialize}{Initialize}
   \SetKwFunction{RunningJob}{RunningJob}
   \SetKwFunction{FreeMachine}{FreeMachine}
   \SetKwBlock{Block}
   \SetAlCapFnt{\footnotesize}
   \KwNotation{\\
	\own{$M$}, \own{$J$} --- the organization owning the processor $M$, the job $J$ \\
    $wait(O)$ --- the set of released, but not-yet scheduled jobs of the organization $O$ at time $t$}

    \vspace{0.3cm}\Initialize{$\mathcal{C}$}:
	\Block{
		\ForEach{$O^{(u)} \in \mathcal{C}$}
		{
			$\textrm{finUt}[O^{(u)}] \leftarrow 0; \; \textrm{finCon}[O^{(u)}] \leftarrow 0$ \;
			$\phi[O^{(u)}] \leftarrow 0; \; \psi[O^{(u)}] \leftarrow 0$ \;
		}
    }
    \vspace{0.3cm}\Schedule{$t_{prev}, t$}: \verb+//+ $t_{prev}$ is the time of the previous event
	\Block{
	  \ForEach{$O^{(u)} \in \mathcal{C}$}
			{
				$\phi[O^{(u)}] \leftarrow \phi[O^{(u)}] + (t-t_{prev}) \textrm{finCon}[O^{(u)}]$\;\nllabel{algline::line1}
				$\psi[O^{(u)}] \leftarrow \psi[O^{(u)}] + (t-t_{prev}) \textrm{finUt}[O^{(u)}]$\;\nllabel{algline::line2}
			}
			$\gamma \leftarrow$ generate a random permutation of the set of all processors\;
			\ForEach{$m \in \gamma$}
			{
				\If{$\textbf{not} \; \FreeMachine{$m$, $t$}$}
				{
					$J \leftarrow$ \RunningJob{$m$}\;
					$\textrm{finUt}[\own{J}] \leftarrow \textrm{finUt}[\own{J}] + t - t_{prev}$ \;
					$\textrm{finCon}[\own{m}] \leftarrow \textrm{finCon}[\own{m}] + t - t_{prev}$ \;

                   			$\phi[\own{J}] \leftarrow \phi[\own{J}] + \frac{1}{2} (t-t_{prev})(t-t_{prev}+1)$\;\nllabel{algline::line3}
                   			$\psi[\own{m}] \leftarrow \psi[\own{m}] + \frac{1}{2} (t-t_{prev})(t-t_{prev}+1)$\;\nllabel{algline::line4}
				}
			}
			\ForEach{$m \in \gamma$}
			{
				\If{\FreeMachine{$m$, $t$} $\;\textbf{\textrm{and}}$ $\;\;\bigcup_{O^{(u)}} wait(O^{(u)}) \neq \emptyset$} 
				{
					$org \leftarrow \textrm{argmax}_{O^{(u)}: wait(O^{(u)}) \neq \emptyset} (\phi[O^{(u)}] - \psi[O^{(u)}])$ \;
					$J \leftarrow$ first waiting job of $org$ \;
					\startJob{$J$, $m$} \;
					$\textrm{finUt}[org] \leftarrow \textrm{finUt}[org] + 1$ \;
					$\textrm{finCon}[\own{m}] \leftarrow \textrm{finCon}[\own{m}] + 1$ \;
				}
			}
	}

\end{algorithm}
\caption{Algorithm~\textsc{DirectContr}: a heuristic algorithm for fair scheduling.}\label{alg:simpleHeuristic}
\end{figure}

\medskip
\noindent
\textbf{\textsc{DirectContr}} (the pseudo-code of the algorithm is given in Figure~\ref{alg:simpleHeuristic}). The algorithm keeps for each organization $O$ its utility $\psi_{sp}[O]$ and its estimated contribution $\phi[O]$. The estimate of the contribution of each organization is assessed directly (without considering any subcoalitions) by the following heuristic. On each scheduling event $t$ we consider the processors in a random order and assign waiting jobs to free processors. The job that is started on processor $m$ increases the contribution $\tilde{\phi}$ of the owner of $m$ by the utility of this job.

In the pseudo code, $\mathrm{finUt}[O]$ denotes the number of the unit-size parts of the jobs of the organization $O$ that are completed before $t_{prev}$. From Equation~\ref{def::fairMetric} we know that the utility in time $t$ of the unit-size parts of the jobs of the organization $O$ that are completed before $t_{prev}$ is greater by $(t - t_{prev})\mathrm{finUt}[O]$ than this utility in time $t_{prev}$ (line~\ref{algline::line1}); the utility of the unit-size parts of the job completed between $t_{prev}$ and $t$ is equal to $\sum_{i = 1}^{t-t_{prev}}i = \frac{1}{2}(t-t_{prev})(t-t_{prev}+1)$ (line~\ref{algline::line3}). Similarly, $\mathrm{finCon}[O]$ denotes the number of the completed unit-size parts of the jobs processed on the processors of the organization $O$. The algorithm updates the utilities and the estimates of the contributions.
The waiting jobs are assigned to the processors in the order of decreasing differences $(\phi - \psi)$ of the issuing organizations (similarly to \textsc{Ref}).

\medskip
\noindent
\textbf{\textsc{RoundRobin}}. The algorithm cycles through the list of organizations to determine the job to be started.

\medskip
\noindent
\textbf{\textsc{FairShare}}~\cite{kay1988fair}. This is perhaps the most popular scheduling algorithm using the idea of distributive fairness. Each organization is given a target weight (a \emph{share}). The algorithm tries to ensure that the resources used by different organizations are proportional to their shares. More formally, whenever there is a free processor and some jobs waiting for execution, the algorithm sorts the organizations in the ascending order of the ratios: the total time of the processor already assigned for the jobs of the organization divided by its share. A job from the organization with the lowest ratio is started.

In all versions of fair share, in the experiments we set the target share to the fraction of  processors contributed by an organization to the global pool.

\medskip
\noindent
\textbf{\textsc{UtFairShare}}.
This algorithm uses the same idea as \textsc{FairShare}. The only difference is that \textsc{UtFairShare} tries to balance the utilities of the organizations instead of their resource allocation. Thus, in each step the job of the organization with the smallest ratio of utility to share is selected. We used this algorithm because it uses the allocation mechanism of \textsc{FairShare}, but operates on the strategy-proof metric used by our referral exponential algorithm.

\medskip
\noindent
\textbf{\textsc{CurrFairShare}}.
This version of the fair share algorithm does not keep any history; it only ensures that, for each organization, the number of currently executing jobs is proportional to its target share. We used this algorithm because, it is light and efficient. It has also an interesting property: the history does not influence the current schedule. We were curious to check how this property influences the fairness.

\subsection{Settings}

To run simulations, we chose the following workloads from the Parallel Workload Archive~\cite{BorjaCite21}:
\begin{inparaenum}
\item LPC-EGEE\footnote{www.cs.huji.ac.il/labs/parallel/workload/l\_lpc/index.html} (cleaned version),
\item PIK-IPLEX\footnote{www.cs.huji.ac.il/labs/parallel/workload/l\_pik\_iplex/index.html},
\item RICC\footnote{www.cs.huji.ac.il/labs/parallel/workload/l\_ricc/index.html},
\item SHARC\-NET-Whale\footnote{www.cs.huji.ac.il/labs/parallel/workload/l\_sharcnet/index.html}.
\end{inparaenum}
We selected traces that closely resemble sequential workloads (in the selected traces most of the jobs require a single processor). We replaced parallel jobs that required $q > 1$ processors with $q$ copies of a sequential job having the same duration.

In each workload, each job has a user identifier (in the workloads there are respectively 56, 225, 176 and 154 distinct user identifiers). To distribute the jobs between the organizations we uniformly distributed the user identifiers between the organizations; the job sent by the given user was assigned to the corresponding organization. 

Because \textsc{Ref} is exponential, the experiments are computationally-intensive; in most of the experiments, we simulate only 5 organizations. 

The users usually send their jobs in consecutive blocks. We also considered a scenario when the jobs are uniformly distributed between organizations (corresponding to a case when the number of users within organizations is large, in which case the distribution of the jobs should be close to uniform). These experiments led to the same conclusions, so we present only the results from the case when the user identifiers were distributed between the organizations.

For each workload, the total number of the processors in the system was equal to the number originally used in the workload (that is 70, 2560, 8192 and 3072, respectively). The processors were assigned to organizations so that the counts follow Zipf and (in different runs) uniform distributions.

For each algorithm, we compared the vector of the utilities (the utilities per organization) at the end of the simulated time period (a fixed time $t_{end}$): $\vec{\psi}$ with the vector of the utilities in the ideally fair schedule $\vec{\psi}^{*}$ (computed by \textsc{Ref}). Let $p_{tot}$ denote the total number of the unit-size parts of the jobs completed in the fair schedule returned by \textsc{Ref}, $p_{tot} = \sum_{(s, p) \in \sigma^{*}: s \leq t_{end}}\min(p, t_{end}-s)$. We calculated the difference $\Delta \psi = \|\vec{\psi} - \vec{\psi}^{*}\| = \sum_{O^{(u)}}(\psi^{(u)} - \psi^{(u), *})$ and compared the values $\Delta \psi / p_{tot}$ for different algorithms. The value $\Delta \psi / p_{tot}$ is the measure of the fairness that has an intuitive interpretation. Since delaying each unit-size part of a job by one time moment decreases the utility of the job owner by one, the value $\Delta \psi / p_{tot}$ gives the average unjustified delay (or unjustified speed-up) of a job due to the unfairness of the algorithm.

\subsection{Results}
We start with experiments on short sub-traces of the original workloads. We randomly selected the start time of the experiment $t_{start}$ and set the end time to $t_{end} = t_{start}+ 5 \cdot 10^4$. For each workload we run 100 experiments (on different periods of workloads of length $5 \cdot 10^4$). The average values of $\Delta \psi / p_{tot}$, and the standard deviations are presented in Table~\ref{tab::differentWorkloads}.

\begin{table}[t]
\centering
\caption{The average delay (or the speed up) of jobs due to the unfairness of the algorithm $\Delta \psi / p_{tot}$ for different algorithms and different workloads. Each row is an average over 100 instances taken as parts of the original workload. The duration of the experiment is $5 \cdot 10^4$.}\label{tab::differentWorkloads}
\footnotesize
\begin{tabular}{c|c|c||c|c||c|c||c|c|}
\cline{2-9}
& \multicolumn{2}{ |c|| }{LPC-EGEE} & \multicolumn{2}{ |c|| }{PIK-IPLEX}  &  \multicolumn{2}{ |c|| }{SHARCNET-Whale} & \multicolumn{2}{ |c| }{RICC}\\ \cline{2-9}
& Avg & St. dev. & Avg & St. dev. & Avg & St. dev.  & Avg & St. dev. \\ \cline{1-9}
\multicolumn{1}{ |l| }{\textsc{RoundRobin}} & 238 & 353 & 6 & 33 & 145 & 38 & 2839 & 357   \\ \cline{1-9}
\multicolumn{1}{ |l| }{\textsc{Rand} ($N = 15$)} & 8 & 21 & 0.014 & 0.01 & 6 & 6 & 162 & 187   \\ \cline{1-9}
\multicolumn{1}{ |l| }{\textsc{DirectContr}} & 5 & 11 & 0.02 & 0.15 & 10 & 7 & 537 & 303  \\ \cline{1-9}
\multicolumn{1}{ |l| }{\textsc{FairShare}} & 16 & 25 & 0.3 & 1.38 & 13 & 8 & 626 & 309   \\ \cline{1-9}
\multicolumn{1}{ |l| }{\textsc{UtFairShare}} & 16 & 25 & 0.3 & 1.38 & 38 & 67 & 515 & 284   \\ \cline{1-9}
\multicolumn{1}{ |l| }{\textsc{CurrFairShare}} & 87 & 106 & 0.3 & 1.58 & 145 & 80 & 1231 & 243   \\ \cline{1-9}
\end{tabular}
\end{table}

From this part of the experiments we conclude that: (i) The algorithm \textsc{Rand} is the most fair algorithm regarding the fairness by the Shapley Value; but \textsc{Rand} is the second most computationally intensive algorithm (after \textsc{Ref}). (ii) All the other algorithms are about equally computationally efficient. The algorithm \textsc{DirectContr} is the most fair. (iii) The algorithm \textsc{FairShare}, which is the algorithm mostly used in real systems, is not much worse than \textsc{DirectContr}. (iv)  Arbitrary scheduling algorithms like \textsc{RoundRobin} may result in unfair schedules. (v) The fairness of the algorithms may depend on the workload. In RICC the differences are much more visible than in PIK-IPLEX. Thus, although \textsc{DirectContr} and \textsc{FairShare} are usually comparable, on some workloads the difference is important. 
%% We claim that this is the reason for using \textsc{DirectContr} instead of \textsc{FairShare} in scheduling.

\begin{table}[tb]
\centering
\caption{The average delay (or the speed up) of jobs due to the unfairness of the algorithm $\Delta \psi / p_{tot}$ for different algorithms and different workloads. Each row is an average over 100 instances taken as parts of the original workload. The duration of the experiment is $5 \cdot 10^5$.}\label{tab::differentWorkloads2}
\footnotesize
\begin{tabular}{c|c|c||c|c||c|c||c|c|}
\cline{2-9}
& \multicolumn{2}{ |c|| }{LPC-EGEE} & \multicolumn{2}{ |c|| }{PIK-IPLEX}  &  \multicolumn{2}{ |c|| }{SHARCNET-Whale} & \multicolumn{2}{ |c| }{RICC}\\ \cline{2-9}
& Avg & St. dev. & Avg & St. dev. & Avg & St. dev.  & Avg & St. dev. \\ \cline{1-9}
\multicolumn{1}{ |l| }{\textsc{RoundRobin}} & 4511 & 6257 & 242 & 1420 & 404 & 1221 & 10850 & 13773 \\ \cline{1-9}
\multicolumn{1}{ |l| }{\textsc{Rand} ($N = 15$)} & 562 & 1670 & 1.3 & 7 & 26 & 158 & 771 & 1479  \\ \cline{1-9}
\multicolumn{1}{ |l| }{\textsc{DirectContr}} & 410 & 1083 & 0.2 & 1.4 & 60 & 204 & 1808 & 3397  \\ \cline{1-9}
\multicolumn{1}{ |l| }{\textsc{FairShare}} & 575 & 1404 & 2.3 & 12 & 94 & 307 & 2746 & 4070   \\ \cline{1-9}
\multicolumn{1}{ |l| }{\textsc{UtFairShare}} & 888 & 2101 & 1.2 & 5 & 120 & 344 & 4963 & 6080   \\ \cline{1-9}
\multicolumn{1}{ |l| }{\textsc{CurrFairShare}} & 1082 & 2091 & 2.2 & 11 & 180 & 805 & 5387 & 9083   \\ \cline{1-9}
\end{tabular}
\end{table}

\begin{figure}[tb]
  \begin{center}
    \includegraphics[scale=0.65]{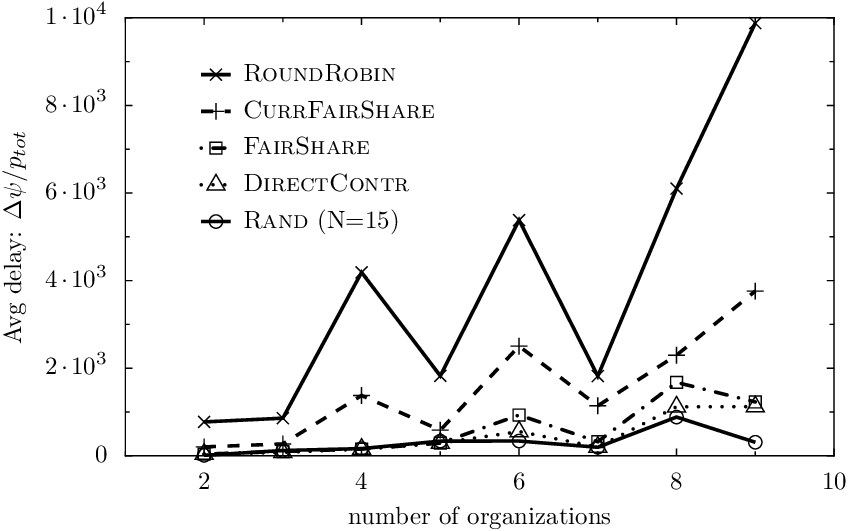}
  \end{center}
  \vspace{-0.5cm}
  \caption{The effect of the number of the organizations on ratio $\Delta \psi / p_{tot}$.}
  \label{fig:manyOrganizations}
\end{figure}

In the second series of experiments, we verified the effect of the duration of the simulated workload on the resulting fairness measure (the ratio $\Delta \psi / p_{tot}$). As we changed the duration of the experiments from $5 \cdot 10^4$ to $5 \cdot 10^5$, we observed that the unfairness ratio $\Delta \psi / p_{tot}$ was increasing. The value of the ratio for $t_{end} - t_{start} = 5 \cdot 10^5$ are presented in Table~\ref{tab::differentWorkloads2}.
The relative quality of the algorithms is the same as in the previous case. Thus, all our previous conclusions hold. However, now all the algorithms are significantly less fair than the exact algorithm. Thus, in long-running systems the difference between the approaches becomes more important.
If there are a few organizations, the exact \textsc{Ref} or the randomized \textsc{Rand} algorithms should be used. In larger systems, when the computational cost of these is too high, \textsc{DirectContr} clearly outperforms \textsc{FairShare}. 
%In case of smaller number of the organizations the referral or the randomized algorithms are a better choice.

Last, we verified the effect of the number of the organizations on the ratio $\Delta \psi / p_{tot}$. The results from the experiments conducted on LPC-EGEE data set are presented in Figure~\ref{fig:manyOrganizations}. As the number of organizations increases, the unfairness ratio $\Delta \psi / p_{tot}$ grows; thus the difference between the algorithms is more significant. This confirms our previous conclusions.
%The relative fairness of the algorithms is the same as in our previous experiments.

\section{Conclusions}
In this paper we define the fairness of the scheduling algorithm in terms of cooperative game theory which allows to quantify the impact of an organization on the utilities of others. We present a non-monetary model in which it is not required that each organization has accurate valuations of its jobs and resources. We show that classic utility functions may create incentives for workload manipulations. We thus propose a strategy resilient utility function that can be thought of as per-organization throughput. 

We analyze the complexity of the fair scheduling problem. The general problem is NP-hard and hard to approximate. Nevertheless, the problem parameterized with the number of organizations is FPT.
Also, the FPT algorithm can be used as a reference for comparing the fairness of different algorithms on small instances (dozens of organizations).
For a special case with unit-size jobs, we propose an FPRAS. In our experiments, we show that the FPRAS can be used as the heuristic algorithm; we also show another efficient heuristic. The main conclusion from the experiments is that in multi-organizational systems, the distributive fairness used by fair share does not result in truly-fair schedules; our heuristic better approximates the Shapley-fair schedules.

We, further, show that every greedy algorithm achieves at least $\frac{3}{4}$-times as good resource utilization as the optimal algorithm. Since this result holds even if the durations and the pattern of incoming jobs are unknown, and if we use arbitrary underlying scheduling policy, we show that the loss of resources utilization due to the fairness is upper-bounded by $25\%$.

Since we do not require the valuation of the jobs, and we consider an on-line, non-clairvoyant scheduling, we believe the presented results have practical consequences for real-life job schedulers.

There are many natural questions for the future work. Although our approach is applicable to the parallel jobs, and to scheduling on related and unrelated machines, we yet do not know the resulting loss of the global efficiency. Determining these bounds is an interesting open question. We suspect that in case of related and unrelated machines the loss of efficiency might be significant. In such case, the next natural question is too look for refinements of our algorithm that would allow to alleviate this problem. Another interesting direction is to explore other game-theoretic notions of fairness.

\bibliographystyle{abbrv}
\bibliography{main}

\begin{thebibliography}{10}

\bibitem{yfqScheduling}
J.~Bruno, J.~Brustoloni, E.~Gabber, B.~Ozden, and A.~Silberschatz.
\newblock Disk scheduling with quality of service guarantees.
\newblock In {\em Proceedings of ICMCS-1999}, page 400, 1999.

\bibitem{buyya05ge}
R.~Buyya, D.~Abramson, and S.~Venugopal.
\newblock The grid economy.
\newblock In {\em Special Issue on Grid Computing}, volume~93, pages 698--714,
  2005.

\bibitem{divisibleLoadSchedCoalitionalGames}
T.~E. Carroll and D.~Grosu.
\newblock Divisible load scheduling: An approach using coalitional games.
\newblock In {\em Proceedings of ISPDC-2007}, pages 36--36, 2007.

\bibitem{fairScheduling2003IEEE}
H.~M. Chaskar and U.~Madhow.
\newblock Fair scheduling with tunable latency: a round-robin approach.
\newblock {\em IEEE/ACM Transactions on Networking}, 11(4):592--601, 2003.

\bibitem{finiteCongestionGames}
G.~Christodoulou and E.~Koutsoupias.
\newblock The price of anarchy of finite congestion games.
\newblock In {\em Proceedings of STOC-2005}, pages 67--73, 2005.

\bibitem{Trystram2010ApproximationAlgorithms}
P.~Dutot, F.~Pascual, K.~Rzadca, and D.~Trystram.
\newblock Approximation algorithms for the multi-organization scheduling
  problem.
\newblock {\em IEEE Transactions on Parallel and Distributed Systems},
  22:1888--1895, 2011.

\bibitem{Ehrgott00approximationalgorithms}
M.~Ehrgott.
\newblock Approximation algorithms for combinatorial multicriteria optimization
  problems.
\newblock {\em International Transactions in Operational Research}, 7:2000,
  2000.

\bibitem{BorjaCite21}
D.~G. Feitelson.
\newblock Parallel workloads archive.
\newblock http://www.cs.huji.ac.il/labs/parallel/workload/.

\bibitem{10.1109/TPDS.2007.34}
P.~Ghosh, K.~Basu, and S.~K. Das.
\newblock A game theory-based pricing strategy to support single/multiclass job
  allocation schemes for bandwidth-constrained distributed computing systems.
\newblock {\em IEEE Transactions on Parallel and Distributed Systems},
  18(3):289--306, 2007.

\bibitem{Ghosh20051366}
P.~Ghosh, N.~Roy, S.~K. Das, and K.~Basu.
\newblock A pricing strategy for job allocation in mobile grids using a
  non-cooperative bargaining theory framework.
\newblock {\em Journal of Parallel and Distributed Computing},
  65(11):1366--1383, 2005.

\bibitem{simpleSfqScheduling}
P.~Goyal, H.~M. Vin, and H.~Chen.
\newblock Start-time fair queueing: a scheduling algorithm for integrated
  services packet switching networks.
\newblock In {\em Proceedings of SIGCOMM-1996}, pages 157--168, 1996.

\bibitem{Grosu02agame-theoretic}
D.~Grosu and A.~T. Chronopoulos.
\newblock A game-theoretic model and algorithm for load balancing in
  distributed systems.
\newblock In {\em Proceedings of IPDPS-2002}, pages 146--153, 2002.

\bibitem{Grosu:2003:TMF:824470.825337}
D.~Grosu and A.~T. Chronopoulos.
\newblock A truthful mechanism for fair load balancing in distributed systems.
\newblock In {\em Proceedings of NCA-2003}, pages 289--289, 2003.

\bibitem{Grosu04auction-basedresource}
D.~Grosu and A.~Das.
\newblock Auction-based resource allocation protocols in grids.
\newblock In {\em Proceedings of PDCS-2004}, pages 20--27, 2004.

\bibitem{Gulati2008}
A.~Gulati and I.~Ahmad.
\newblock Towards distributed storage resource management using flow control.
\newblock {\em ACM SIGOPS Operating Systems Review}, 42(6):10--16, 2008.

\bibitem{parda2009Fast}
A.~Gulati, I.~Ahmad, and C.~A. Waldspurger.
\newblock {PARDA: Proportional Allocation of Resources for Distributed Storage
  Access}.
\newblock In {\em Proceedings of FAST-2009}, February 2009.

\bibitem{Hasan:2005:SPS:1058431.1059197}
R.~Hasan, Z.~Anwar, W.~Yurcik, L.~Brumbaugh, and R.~Campbell.
\newblock A survey of peer-to-peer storage techniques for distributed file
  systems.
\newblock In {\em Proceedings of ITCC-2005}, pages 205--213, 2005.

\bibitem{multiDimensionalScheduling}
L.~Huang, G.~Peng, and T.~Chiueh.
\newblock Multi-dimensional storage virtualization.
\newblock {\em ACM SIGMETRICS Performance Evaluation Review}, 32(1):14--24,
  2004.

\bibitem{Inoie04paretoset}
A.~Inoie, H.~Kameda, and C.~Touati.
\newblock Pareto set, fairness, and nash equilibrium: A case study on load
  balancing.
\newblock In {\em Proceedings of ISDG-2004}, pages 386--393, 2004.

\bibitem{Jackson:2001:CAM:646382.689682}
D.~B. Jackson, Q.~Snell, and M.~J. Clement.
\newblock Core algorithms of the maui scheduler.
\newblock In {\em Proceedings of JSSPP-2001}, pages 87--102, 2001.

\bibitem{fsfqScheduling}
W.~Jin, J.~S. Chase, and J.~Kaur.
\newblock Interposed proportional sharing for a storage service utility.
\newblock {\em ACM SIGMETRICS Performance Evaluation Review}, 32(1):37--48,
  2004.

\bibitem{Triage2005ACM}
M.~Karlsson, C.~Karamanolis, and X.~Zhu.
\newblock Triage: Performance differentiation for storage systems using
  adaptive control.
\newblock {\em ACM Transactions on Storage}, 1(4):457--480, 2005.

\bibitem{kay1988fair}
J.~Kay and P.~Lauder.
\newblock A fair share scheduler.
\newblock {\em Communications of the ACM}, 31(1):44--55, 1988.

\bibitem{kenyon04gr}
C.~Kenyon and G.~Cheliotis.
\newblock Grid resource commercialization: economic engineering and delivery
  scenarios.
\newblock In J.~Nabrzyski, J.~M. Schopf, and J.~Weglarz, editors, {\em Grid
  resource management: state of the art and future trends}, pages 465--478.
  Kluwer Academic Publishers, Norwell, MA, USA, 2004.

\bibitem{journals/eor/KostrevaOW04}
M.~M. Kostreva, W.~Ogryczak, and A.~Wierzbicki.
\newblock Equitable aggregations and multiple criteria analysis.
\newblock {\em European Journal of Operational Research}, 158(2):362--377,
  2004.

\bibitem{lee2006user}
C.~B. Lee and A.~Snavely.
\newblock On the user--scheduler dialogue: studies of user-provided runtime
  estimates and utility functions.
\newblock {\em International Journal of High Performance Computing
  Applications}, 20(4):495--506, 2006.

\bibitem{conf/cocoon/Liben-NowellSWW12}
D.~Liben-Nowell, A.~Sharp, T.~Wexler, and K.~Woods.
\newblock Computing shapley value in supermodular coalitional games.
\newblock In {\em Proceedings of COCOON-2012}, volume 7434, pages 568--579,
  2012.

\bibitem{Mashayekhy:2011:MMD:2358951.2359134}
L.~Mashayekhy and D.~Grosu.
\newblock A merge-and-split mechanism for dynamic virtual organization
  formation in grids.
\newblock In {\em Proceedings of PCCC-2011}, PCCC '11, pages 1--8, 2011.

\bibitem{costSharingInJobSched}
D.~Mishra and B.~Rangarajan.
\newblock Cost sharing in a job scheduling problem using the shapley value.
\newblock In {\em Proceedings of ACM-EC-2005}, pages 232--239, 2005.

\bibitem{Moulin:2007:SFP:1527888.1527890}
H.~Moulin.
\newblock On scheduling fees to prevent merging, splitting, and transferring of
  jobs.
\newblock {\em Mathematics of Operations Research}, 32(2):266--283, May 2007.

\bibitem{vocking-selfishloadbalancing}
N.~Nisan, T.~Roughgarden, E.~Tardos, and V.~V. Vazirani.
\newblock {\em Algorithmic Game Theory}, chapter Selfish Load Balancing.
\newblock Cambridge University Press, 2007.

\bibitem{routingGames}
N.~Nisan, T.~Roughgarden, E.~Tardos, and V.~V. Vazirani.
\newblock {\em Algorithmic Game Theory}, chapter Routing Games.
\newblock Cambridge University Press, 2007.

\bibitem{RePEc:mtp:titles:0262650401}
J.~M. Osborne and A.~Rubinstein.
\newblock {\em A Course in Game Theory}, volume~1 of {\em MIT Press Books}.
\newblock The MIT Press, 1994.

\bibitem{Penmasta:2006:PUJ:1898699.1898880}
S.~Penmasta and A.~T. Chronopoulos.
\newblock Price-based user-optimal job allocation scheme for grid systems.
\newblock In {\em Proceedings of IPDPS-2006}, pages 336--336, 2006.

\bibitem{Rahman:2011:DSA:2018436.2018458}
R.~Rahman, T.~Vink\'{o}, D.~Hales, J.~Pouwelse, and H.~Sips.
\newblock Design space analysis for modeling incentives in distributed systems.
\newblock In {\em Proceedings of SIGCOMM-2011}, pages 182--193, 2011.

\bibitem{linearCongestionGames}
A.~Roth.
\newblock The price of malice in linear congestion games.
\newblock In {\em Proceedings of WINE-2008}, pages 118--125, 2008.

\bibitem{Rzadca2007FairGameTheoreticResource}
K.~Rzadca, D.~Trystram, and A.~Wierzbicki.
\newblock Fair game-theoretic resource management in dedicated grids.
\newblock In {\em Proceedings of CCGRID-2007}, 2007.

\bibitem{shapley_53}
L.~S. Shapley.
\newblock {A value for n-person games}.
\newblock {\em Contributions to the theory of games}, 2:307--317, 1953.

\bibitem{citeulike:1087789}
K.~D. Vohs, N.~L. Mead, and M.~R. Goode.
\newblock {The Psychological Consequences of Money}.
\newblock {\em Science}, 314(5802):1154--1156, Nov. 2006.

\bibitem{proportionalScheduling2007Fast}
Y.~Wang and A.~Merchant.
\newblock Proportional-share scheduling for distributed storage systems.
\newblock In {\em Proceedings of FAST-2007}, pages 4--4, 2007.

\bibitem{adaptiveOverload2003USENIX}
M.~Welsh and D.~Culler.
\newblock Adaptive overload control for busy internet servers.
\newblock In {\em Proceedings of USITS-2003}, pages 4--4, 2003.

\bibitem{Wierman}
A.~Wierman.
\newblock {Fairness and scheduling in single server queues}.
\newblock {\em Surveys in Operations Research and Management Science},
  16:39--48, 2011.

\bibitem{aumann_shapley}
E.~Winter.
\newblock {The Shapley value}.
\newblock In R.~Aumann and S.~Hart, editors, {\em Handbook of Game Theory with
  Economic Applications}, volume~3, pages 2025--2054. 2002.

\bibitem{Yaiche:2000:GTF:355154.355166}
H.~Ya\"{\i}che, R.~R. Mazumdar, and C.~Rosenberg.
\newblock A game theoretic framework for bandwidth allocation and pricing in
  broadband networks.
\newblock {\em IEEE/ACM Transactions on Networking}, 8(5):667--678, Oct. 2000.

\bibitem{2levelIOScheduling}
J.~Zhang, A.~Sivasubramaniam, A.~Riska, Q.~Wang, and E.~Riedel.
\newblock An interposed 2-level i/o scheduling framework for performance
  virtualization.
\newblock In {\em Proceedings of SIGMETRICS-2005}, pages 406--407, 2005.

\bibitem{generalAboutControlTheory}
X.~Zhu, M.~Uysal, Z.~Wang, S.~Singhal, A.~Merchant, P.~Padala, and K.~Shin.
\newblock What does control theory bring to systems research?
\newblock volume~43, pages 62--69, 2009.

\end{thebibliography}

\end{document}